\documentclass[11pt,twoside,a4paper]{article}

\usepackage{algpseudocode}
\usepackage{amsmath}

\usepackage{amsfonts}
\usepackage{amssymb}
\usepackage{amsthm}
\usepackage{array}
\usepackage{BeamerColor}
\usepackage{caption}
\usepackage{subcaption}
\usepackage{cite}
\usepackage{color}
\usepackage{dsfont}
\usepackage{empheq}
\usepackage{hyperref}
\usepackage{listings}
\usepackage{multicol}
\usepackage{paralist}
\usepackage{relsize}
\usepackage{stmaryrd}
\usepackage{tabularx}
\usepackage{tikz}
\usetikzlibrary{calc,arrows,positioning}
\usepackage{url}
\usepackage[]{units}
\usepackage{verbatim}
\usepackage{wasysym}

\allowdisplaybreaks

\newtheorem{theorem}{Theorem}
\newtheorem{proposition}{Proposition}

\newtheorem{corollary}{Corollary}

\theoremstyle{definition}
\newtheorem{definition}{Definition}
\newtheorem{example}{\emph{Example}}
\let\doendexample\endexample
\renewcommand\endexample{~\hfill\scalebox{0.85}{$\LHD$}\doendexample} 

\newtheorem*{notation}{Notation}

\usepackage{thmtools}
\usepackage{thm-restate}

\newcommand{\logicName}{Evolution Temporal Logic}
\newcommand{\logicShort}{EvTL}
\newcommand{\logicSymbol}{\mathcal L}
\newcommand{\logicMetric}{\ell_{\rho,\OT}}
\newcommand{\system}{\mathbf{s}}
\newcommand{\System}{\mathbf{S}_\D}

\newcommand{\sinistra}[1]{\textbf{target}(#1)}
\newcommand{\destra}[1]{\textbf{brink}(#1)}

\newcommand{\cstep}{\mathsf{step}}

\newcommand{\fevent}[2]{\Diamond^{#1} #2}
\newcommand{\fglob}[2]{\Box^{#1} #2}
\newcommand{\funtil}[3]{#1~\mathcal{U}^{#2}~#3}

\newcommand{\nats}{\mathbb{N}}
\newcommand{\borel}{{\cal B}}

\newcommand{\sat}[2]{\llbracket #1 \rrbracket_{#2}}

\newcommand{\D}{\mathcal{D}}
\newcommand{\ds}{\mathbf{d}}

\newcommand{\mesA}{\mathbb{A}}
\newcommand{\mesB}{\mathbb{B}}

\newcommand{\mesD}{\mathbb{D}}

\newcommand{\distrib}{\Delta}

\newcommand{\PP}{\mathit{Pr}}

\newcommand{\Var}{\mathrm{Var}}

\DeclareMathOperator{\Wasserstein}{\mathbf{W}}

\newcommand{\dirac}{\delta}

\newcommand{\OT}{\mathrm{OT}}

\newcommand{\traccione}{evolution sequence}
\newcommand{\tracciones}{evolution sequences}
\newcommand{\Traccione}{Evolution sequence}

\newcommand{\spell}{evolution }

\newcommand{\dataspace}{data space}

\newcommand{\Dataspace}{Data space}

\newcommand{\datastate}{data state}
\newcommand{\datastates}{data states}
\newcommand{\Datastate}{Data state}

\newcommand{\ES}{\mathcal{S}}

\newcommand{\dd}{\mathrm{d}}

\newcommand{\F}{\Sigma}

\newcommand{\m}{\mathfrak{m}}

\newcommand{\W}{\mathfrak{W}}
\newcommand{\w}{\mathfrak{w}}

\newcommand{\real}{\mathbb{R}} 

\newcommand{\var}[1]{\mathrm{var}(#1)}

\algnewcommand\algorithmicswitch{\textbf{match}}
\algnewcommand\algorithmiccase{\textbf{with}}
\algnewcommand\algorithmicassert{\texttt{assert}}
\algnewcommand\Assert[1]{\State \algorithmicassert(#1)}%
\algdef{SE}[SWITCH]{Switch}{EndSwitch}[1]{\algorithmicswitch\ #1\ }{\algorithmicend\ \algorithmicswitch}%
\algdef{SE}[CASE]{Case}{EndCase}[1]{\algorithmiccase\ #1\ :}{\algorithmicend\ \algorithmiccase}%
\algtext*{EndSwitch}%
\algtext*{EndCase}%

\begin{document}

\title{EvTL: A Temporal Logic for the Transient Analysis of Cyber-Physical Systems}
\author{
Valentina Castiglioni
\and
Michele Loreti
\and 
Simone Tini
}

\date{}

\maketitle

\begin{abstract}
The behaviour of systems characterised by a closed interaction of software components with the environment is inevitably subject to perturbations and uncertainties.
In this paper we propose a general framework for the specification and verification of requirements on the behaviour of these systems.
We introduce the \emph{Evolution Temporal Logic} (\emph{EvTL}), a stochastic extension of STL allowing us to specify properties of the probability distributions describing the transient behaviour of systems, and to include the presence of uncertainties in the specification.
We equip EvTL with a robustness semantics and we prove it sound and complete with respect to the semantics induced by the \emph{evolution metric}, i.e., a hemimetric expressing how well a system is fulfilling its tasks with respect to another one.
Finally, we develop a \emph{statistical model checking algorithm} for EvTL specifications.
As an example of an application of our framework, we consider a three-tanks laboratory experiment.
\end{abstract}


\section{Introduction}
\label{sec:introduction}

Cyber-physical systems~\cite{CPS-DEF}, IoT systems~\cite{K11} and smart devices are characterised by 
software applications that must be able to deal with \emph{highly changing operational conditions}, henceforth referred to as the \emph{environment}.
Examples of these applications are the software components of unmanned vehicles, controllers, (on-line) service applications, the devices in a smart house, etc.
In these contexts, the \emph{behaviour} of a system is the result of the \emph{interplay of the devices, or software components, with their environment}.

The main challenge in the analysis and verification of these systems is then the dynamical and, sometimes, unpredictable behaviour of the environment.
The highly dynamic behaviour of physical processes can only be approximated in order to become computationally tractable and can constitute a safety hazard for the devices in the system (like, e.g., an unexpected gust of wind for a drone that is autonomously setting its trajectory to avoid obstacles); some devices may appear, disappear, or become temporarily unavailable; faults or conflicts may occur (like, e.g., in a smart home the application responsible for the ventilation of a room may open a window in conflict with the one that has to limit the noise level); sensors may introduce measurement errors; etc.
Introducing \emph{uncertainties} and \emph{approximations} in these systems is therefore inevitable to achieve some degree of system robustness.

Clearly, this uncertain, stochastic, behaviour needs to be taken into account in the specification, and the consequent verification, of requirements over systems.
\begin{quote}
\emph{The main objective of this paper is then to provide a general framework to model and check properties of systems running under uncertainties.}
\end{quote}


\subsection*{The \traccione.}

Our starting point is the observation that the behaviour of the systems we study can be modelled in a purely \emph{data-driven} fashion:
while the environmental conditions are (partially) available to the software components as a set of data (for instance collected by sensors), the latter ones can in turn use data (for instance communicated through actuators) to (partially) control the environment and fulfil their tasks.
Hence, it is natural to adopt the (discrete time) model of \cite{CLT21} and represent the software-environment interplay in terms of the changes they induce on a set of application-relevant data, henceforth referred to as the \emph{\dataspace}.
Let us call \emph{\datastate{}} the description of the current state of the \dataspace.
Following \cite{CLT21}, at each step, both the software component and the environment induce some changes on the \datastate, providing thus a new \datastate{} at the next step.
This is an abstraction: in concrete, the application operates at a given frequency, thus changing the values in the \datastate{} at each time tick, while the environment modifies continuously these values between two ticks. 
We focus on the sum of the effects of the actions of both components.
However, changes on data are subject to the presence of uncertainties, meaning that it is not always possible to determine exactly the values assumed by data at the next step.
Therefore, we represent the changes induced at each step as a \emph{probability measure} on the attainable \datastates.
For instance, we can assume the computation steps of the system to be determined by a \emph{Markov kernel}.
The behaviour of the system is then entirely expressed by its \emph{\traccione}, i.e., the sequence of probability measures over the \datastates{} obtained at each step.
We remark that the \traccione{} of a system takes into account the effect of perturbations and uncertainties at each time step.
Therefore, by expressing requirements on the \traccione{} of a system we are able to verify the overall behaviour, as well as properties of the step-by-step behaviour.


\subsection*{A novel temporal logic.}

In the literature, \emph{quantitative extensions of model checking} have been proposed, like \emph{stochastic} (or \emph{probabilistic}) model checking \cite{BdAFK18,Bai16,KNP07,KP12}, and \emph{statistical model checking} \cite{SVA04,SVA05,ZPC13,HHA17,BMS16}.
These techniques rely either on a full specification of the system to be checked, or on the possibility of simulating the system by means of a Markovian model or Bayesian inference on samples.
Then, quantitative model checking is based on a specification of requirements in a \emph{probabilistic temporal logic}, such as PCTL \cite{HJ94}, CSL \cite{ASSB00,ASSB96}, probabilistic variants of LTL \cite{Pnu77}, etc.
Similarly, if Runtime Verification \cite{BFFR18} is preferred to off-line verification, probabilistic variants of MTL \cite{Koy90} and STL \cite{MN04} were proposed \cite{TH16,SK16}.
In the quantitative setting, uncertainties are usually dealt with in temporal logics by imposing \emph{probabilistic guarantees on a given property} to be satisfied:
All the aforementioned logics provide a quantitative construct of the form $\varphi_{\bowtie p}$ where $\varphi$ is a formula (the property), $p \in [0,1]$ is a threshold (the probabilistic guarantee), and $\bowtie \in \{<, \le, \ge, >\}$.
A system $\system$ satisfies $\varphi_{\bowtie p}$ if the total probability mass of the runs of $\system$ satisfying $\varphi$ is $\bowtie p$.

This approach is natural and has found several applications, like establishing formal guarantees on reachability.
However, \emph{it does not allow us to analyse the properties of the distributions describing the transient behaviour of the system}.
To specify complex system requirements, under uncertainties, we need to be able to characterise the distribution on data at a given time.
For instance, in the probabilistic risk assessment analysis of the decommissioning of a nuclear power plant, one of the main concerns is related to the likelihood of the deflagration of hydrogen \cite{MAW16}.
In detail, the objective there is to estimate the probability of an overpressurization failure of the reactor building given hydrogen deflagration.
This estimation follows from the \emph{comparison of the probability distribution of the pressure generated by the deflagration with the probability distribution of the pressure resistance of the reactor building}.
Informally, the area of the overlapping region between the two curves gives the desired estimation of the failure probability.
We remark that classic temporal logics, like PCTL, would allow us to verify whether the probability that the pressure generated by the deflagration (or, respectively, the probability of the pressure resistance) is within a given interval. 
However, they do not allow us to verify, in practice, whether the values of such pressure (or, respectively, resistance) are distributed according to a specific distribution.
In the situation described above, this disparity is crucial, since the risk assessment can only be carried out with that information.

To capture this kind of properties, we introduce the \emph{\logicName{}} (\emph{\logicShort}) as a probabilistic variant of STL characterised by the use of stochastic signals: the probabilistic operator $\varphi_{\bowtie p}$ is replaced by atomic propositions being probability measures over \datastates.
Intuitively, by modelling the evolution in time of the probability measures over data, we can gain useful information on the behaviour of the system, including its transient behaviour, and we can also express the presence of uncertainties explicitly in the formulae.


\subsection*{\logicShort{} robustness.}

We equip \logicShort{} with a real-valued semantics expressing the \emph{robustness} of the satisfaction of \logicShort{} specifications.
The \emph{robustness} of a system $\system$ with respect to a formula $\varphi$ is expressed as a real number $\sat{\varphi}{\system} \in [-1,1]$: if it is positive, $\system$ satisfies $\varphi$.
In detail, $\sat{\varphi}{\system}$ describes how much the behaviour of $\system$ has to be modified in order to violate (or satisfy) $\varphi$.
We can then interpret $\sat{\varphi}{\system}$ as an indicator of \emph{how well} $\system$ behaves with respect to the requirement $\varphi$.
Hence, the challenge is to properly formalise ``\emph{how well}".

To this end, we use the \emph{\spell metric} of \cite{CLT21}, a (time-dependent) hemimetric on the \tracciones{} of systems based on a \emph{hemimetric on \datastates{}} and the \emph{Wasserstein metric} \cite{W69}.
The former is defined in terms of a \emph{penalty function} allowing us to compare two \datastates{} only on the base of the objectives of the system. 
The latter lifts the hemimetric on \datastates{} to a hemimetric on probability measures on \datastates.
We then obtain a hemimetric on \tracciones{} 
as the maximum of the Wasserstein distances over time.
The reason to opt for a hemimetric, instead of a more standard (pseudo)metric, is that it allows us to compare the \emph{relative behaviour} of two systems and thus to express whether one system is \emph{better than} the other.
We use the \spell metric to define the robustness of systems with respect to \logicShort{} formulae.
As atomic propositions are probability measures over \datastates, by means of the \spell metric we can directly compare them to the probability measures in the \tracciones{} of systems. 
In this way, we obtain useful information on the differences in the behaviour of two systems from the comparison of their robustness. 
In particular, we prove the robustness to be \emph{sound} and \emph{complete} with respect to our metric semantics: whenever the robustness of $\system_1$ with respect to a formula $\varphi$ is greater than the distance between $\system_1$ and $\system_2$, then we can conclude that the robustness of $\system_2$ with respect to $\varphi$ is positive.

Finally, we provide a \emph{statistical model checking algorithm} for the verification of \logicShort{} specifications, consisting of three components:
\begin{inparaenum}
\item A simulation procedure for the \traccione{} of a system.
\item An algorithm, based on statistical inference, for the evaluation of the Wasserstein distance over probability measures.
\item A procedure that computes the robustness with respect to a formula $\varphi$, by inspecting its syntax.
\end{inparaenum}

In order to show how our techniques can be applied, we consider a very classical problem, namely the $n$-tanks experiment.
Several variants (with different number of tanks) of this problem have been widely used in control program education
(see, among the others, \cite{AL92,RKOMI97,Johansson00,ALVARADO2006}).
Moreover, some recently proposed cyber-physical security testbeds, like \emph{SWaT}~\cite{MT16,Antonioli2017}, can be considered as an evolution of the tanks experiment.  
Here, we consider a variant of the three-tanks laboratory experiment 
described in \cite{RKOMI97}.
We plan to tackle in the future more complex case studies, like the SWaT of \cite{MT16,Antonioli2017}, and the risk assessment analysis of the decommissioning of a nuclear power plant explained above.


\subsection*{Summary of contributions.}

Our main contributions can be summarised as follows:
\begin{enumerate}
\item We introduce the \emph{\logicName{}} (\emph{\logicShort}), a probabilistic variant of STL allowing us to express requirements on systems under uncertainties.
By means of \logicShort{} we can capture the properties of the transient probabilities of systems, and we can express explicitly the presence of uncertainties in the specifications.
\item We use the \spell metric of \cite{CLT21} to define the \emph{robustness} of \logicShort{} specifications, which we prove to be sound and complete with respect to our metric semantics.
\item We provide a statistical model checking algorithm for \logicShort{} specifications.
\item To show the adequacy of our approach we apply it to the three-tanks experiment.
\end{enumerate}
The technical proofs and the simplest parts of the algorithm can be found in the Appendix.


\section{Background}
\label{sec:background}


\paragraph{Measurable spaces}
A \emph{$\sigma$-algebra} over a set $\Omega$ is a family $\F$ of subsets of $\Omega$ s.t.\ $\Omega \in \F$ and $\F$ is closed under complementation and under countable union. 
The pair $(\Omega, \Sigma)$ is called a \emph{measurable space} and the sets in $\Sigma$ are called \emph{measurable sets}, ranged over by $\mesA,\mesB,\dots$.
For an arbitrary family $\Psi$ of subsets of $\Omega$, the $\sigma$-algebra \emph{generated} by $\Psi$ is the smallest $\sigma$-algebra over $\Omega$ containing $\Psi$.
In particular, given a topology $T$ over $\Omega$, the \emph{Borel $\sigma$-algebra} over $\Omega$, denoted $\borel(\Omega)$, is the $\sigma$-algebra generated by the open sets in $T$.
Given two measurable spaces $(\Omega_i,\Sigma_i)$, $i = 1,2$, the \emph{product $\sigma$-algebra} $\Sigma_1 \otimes \Sigma_2$ is the $\sigma$-algebra on $\Omega_1 \times \Omega_2$ generated by the sets $\{\mesA_1 \times \mesA_2 \mid \mesA_i \in \Sigma_i\}$.
In particular, for any $n \in N$, if $\Omega_1, \dots, \Omega_n$ are Polish spaces, then $\borel( \bigtimes_{i=1}^n \Omega_i) = \bigotimes_{i=1}^n \borel(\Omega_i)$ \cite{Bogachev}.


\paragraph{Distributions}
On a measurable space $(\Omega,\Sigma)$, a function $\mu \colon \Sigma \to [0,1]$ is a \emph{probability measure} if 
$\mu(\Omega) = 1$,
$\mu(\mesA) \ge 0$ for all $\mesA \in \Sigma$,
$\mu( \bigcup_{i \in I} \mesA_i) = \sum_{i \in I}\mu(\mesA_i)$ for every countable family of pairwise disjoint measurable sets $\{\mesA_i\}_{i\in I} \subseteq \Sigma$.
With a slight abuse of terminology, we shall use the term \emph{distribution} in place of probability measure.

We let $\distrib{(\Omega,\F)}$ denote the set of all distributions over $(\Omega,\F)$.
For $\omega \in \Omega$, the \emph{Dirac distribution} $\dirac_{\omega}$ is defined by $\dirac_\omega (\mesA) = 1$, if $\omega \in \mesA$, and $\dirac_{\omega}(\mesA) = 0$, otherwise, for all $\mesA \in \F$.
Given a countable set $(p_i)_{i \in I} \in \real$ with $p_i \ge 0$ and $\sum_{i \in I}p_i = 1$, the \emph{convex combination} of the distributions $\{\mu_i\}_{i \in I}$ is the distribution $\sum_{i \in I} p_i \cdot \mu_i$ defined by $(\sum_{i \in I} p_i \cdot \mu_i)(\mesA) = \sum_{i \in I} p_i \mu_i(\mesA)$, for all $\mesA \in \F$.

Assume measurable spaces $(\Omega,\F)$, $(\Omega',\F')$ and $\mu \in \distrib{(\Omega,\F)}$.
Then, $X \colon \Omega \to \Omega'$ is a \emph{random variable} if it is $\F$-measurable, i.e., $X^{-1}(\mesA) \in \F$ for all $\mesA \in \F'$.
The \emph{distribution measure} of $X$ is the distribution $\mu_X$ on $(\Omega',\F')$ defined 
by $\mu_X(\mesA)=\mu(X^{-1}(\mesA))$ for all $\mesA \in \F'$.
We write $X \sim \mu_X$ if $X$ has $\mu_X$ as distribution measure.


\paragraph{The Wasserstein hemimetric}

A \emph{metric} on a set $\Omega$ is a function $m \colon \Omega \times \Omega \to \real^{\ge0}$ s.t.\ $m(\omega_1,\omega_2) = 0$ if{f} $\omega_1 = \omega_2$, $m(\omega_1,\omega_2) = m(\omega_2,\omega_1)$, and $m(\omega_1,\omega_2) \le m(\omega_1,\omega_3) + m(\omega_3,\omega_2)$, for all $\omega_1,\omega_2,\omega_3 \in \Omega$.
We obtain a \emph{hemimetric} by relaxing the first property to $m(\omega_1,\omega_2) = 0$ if $\omega_1=\omega_2$, and by dropping the requirement on symmetry. 
A (hemi)metric $m$ is $l$-\emph{bounded} if $m(\omega_1,\omega_2) \le l$ for all $\omega_1,\omega_2 \in\Omega$.

In this paper we are interested in defining a \emph{hemimetric on distributions}.
To this end we will make use of the Wasserstein lifting \cite{W69} 
which is well defined on Polish spaces equipped with the Borel $\sigma$-algebra (see e.g. \cite{Vil08}).

\begin{definition}
[Wasserstein hemimetric]
\label{def:Wasserstein}
Consider a Polish space $\Omega$ and let $m$ be a hemimetric on $\Omega$.
For any two distributions $\mu$ and $\nu$ on $(\Omega,\borel(\Omega))$, the \emph{Wasserstein lifting} of $m$ to a distance between $\mu$ and $\nu$ is defined by
\[
\Wasserstein(m)(\mu,\nu) = \inf_{\w \in \W(\mu,\nu)} \int_{\Omega \times \Omega} m(\omega,\omega') \dd\w(\omega,\omega')
\]
where $\W(\mu,\nu)$ is the set of the \emph{couplings of $\mu$ and $\nu$}, namely the set of distributions $\w$ over the product space $(\Omega \times \Omega, \borel(\Omega \times \Omega))$ having $\mu$ and $\nu$ as left and right marginal, respectively, i.e., $\w(\mesA, \Omega) = \mu(\mesA)$ and $\w(\Omega, \mesA) = \nu(\mesA)$, for all $\mesA \in \borel(\Omega)$.
\end{definition}

Despite the original version of the Wasserstein distance being defined on a metric on $\Omega$, the Wasserstein hemimetric given above is well-defined.
This is proved in \cite{FR18}.
In particular, the Wasserstein hemimetric is given in \cite{FR18} as Definition 7 (considering the compound risk excess metric as in Equation (31)), and Proposition 4 in \cite{FR18} guarantees that it is indeed a well-defined hemimetric on $\distrib(\Omega,\borel(\Omega))$.
Moreover, Proposition 6 in \cite{FR18} guarantees that the same result holds for the hemimetric $m(x,y) = \max\{x-y,0\}$, which will play an important role in our work (cf.\ Definition~\ref{def:metric_DS} below).

As elsewhere in the literature, we shall henceforth use the term \emph{metric} in place of the term hemimetric.


\section{The model}
\label{sec:calcolo}

Following~\cite{CLT21}, we describe the behaviour of a system in terms of a \emph{probabilistic evolution of data}. 
This is fully motivated in various contexts.
For instance, in a cyber-physical system, the interaction between the logic component and the physical one can be naturally described by focusing on the values that are assumed by physical quantities, and on those that are detected by sensors and assigned to actuators.
Moreover, one introduces \emph{probability} as an abstraction mechanism in order to average over the effect of inessential or unknown details of the evolution of physical quantities which may be also impossible to observe in practice.
Probability also allows one to quantify the degree of approximation introduced by some instruments, such as the sensors.

Technically, we assume a \emph{\dataspace{}} defined by means of a \emph{finite} set of \emph{variables} $\mathrm{Var}$ representing:\ 
\begin{inparaenum}[i)]
\item \emph{environmental conditions}, such as pressure, temperature, humidity, etc.,
\item \emph{values perceived by sensors}, which depend on the value of environmental conditions and are unavoidably affected by imprecision and approximations introduced by sensors, and 
\item \emph{state of actuators}, which are usually elements in a discrete domain, like $\{\mathit{on},\mathit{off}\}$. 
\end{inparaenum}
Without loss of generality, we assume that 
for each $x \in \mathrm{Var}$ the domain $\D_x \subseteq \real$ is either a \emph{finite set} or a \emph{compact} subset of $\real$.
Notice that, in particular, this means that $\D_x$ is a Polish space.
Moreover, as a $\sigma$-algebra over $\D_x$ we assume the Borel $\sigma$-algebra, denoted $\borel_x$.
As $\Var$ is a finite set, we can always assume it to be ordered, namely $\Var=\{x_1,\dots,x_{n}\}$ for a suitable $n \in \nats$.

\begin{definition}
[\Dataspace]
We define the \emph{\dataspace{}} over $\mathrm{Var}$, notation $\D_{\Var}$, as the Cartesian product of the variables domains, namely $\D_{\Var} = \bigtimes_{i = 1}^n \D_{x_i}$.
Then, as a $\sigma$-algebra on $\D_\Var$ we consider the product $\sigma$-algebra $\borel_{\D_\Var} = \bigotimes_{i=1}^n \borel_{x_i}$.
\end{definition}

When no confusion arises, we use $\D$  for $\D_{\mathrm{Var}}$ and $\borel_\D$ for $\borel_{\D_{\mathrm{Var}}}$.
Then, we let $\System$ denote the set of systems having $\D$ as \dataspace.
Elements in $\D$ are the $n$-ples of the form $(v_1,\dots,v_n)$, with $v_i \in \D_{x_i}$, which can be also identified by means of functions $\ds \colon \mathrm{Var} \to \real$ from variables to values, with $\ds(x) \in \D_x$ for all $x \in \mathrm{Var}$. 
Each function $\ds$ identifies a particular configuration in the \dataspace, and it is thus called a \emph{\datastate}.

\begin{definition}
[\Datastate]
\label{def:datastate}
A \emph{\datastate{}} is a mapping $\ds \colon \mathrm{Var} \to \real $ from variables to values, with $\ds(x) \in \D_x$ for all $x \in \mathrm{Var}$.  
\end{definition}

We define an \emph{\traccione{}} as a sequence of  distributions over \datastates{} describing the dynamics of a system.
This sequence is countable as we adopt a discrete time approach.
In this paper we do not focus on how evolution sequences are generated: we simply assume a function $\cstep \colon \D \to \distrib(\D,\borel_\D)$ governing the evolution of the system.
In particular, we recall that, at each step, the activity of the software component depends on the available data and environment conditions, and potential modifications by the component to these data may trigger different behaviours of the environment. 
Hence, it is reasonable to assume that the evolution at each step depends only on the current state.
Consequently, we can assume that $\cstep$ is a \emph{Markov kernel} and that our evolution sequence is the \emph{Markov process} generated by $\cstep$ (see Definition~\ref{def:traccione} below).
Formally, $\cstep(\ds)(\mesD)$ expresses the probability of a system $\system$ to reach a \datastate{} in $\mesD$ from the \datastate{} $\ds$ in one computation step.
Clearly, each system is characterised by a particular function $\cstep$.
Moreover, it is also natural to assume that each system will start its computation from a determined configuration.
Hence, for each system $\system \in \System$, we let $\ds_\system$ denote the \datastate{} from which $\system$ starts its computation.

\begin{definition}
[\Traccione]
\label{def:traccione}
Assume a Markov kernel $\cstep \colon \D \to \distrib(\D,\borel_\D)$ generating the behaviour of system $\system$.
Then, the \emph{\traccione{}} of $\system$ is a countable sequence of distributions in $\distrib(\D,\borel_\D)$ of the form $\ES_{\system} = \ES_{\system, 0} \dots \ES_{\system, n} \dots$ such that, for all $\mesD \in \borel_{\D}$: 
\begin{align*}
& \ES_{\system,0}(\mesD)  = \dirac_{\ds_\system}(\mesD)
\\
& \ES_{\system, i+1}(\mesD)  =  \int_{\D} \cstep(\ds)(\mesD) \; \dd(\ES_{\system, i}(\ds)) 
\enspace .
\end{align*}
\end{definition}

For a possible definition of $\cstep$, we refer to~\cite{CLT21}.
There, a (cyber-physical) system is specified as a combination of a \emph{program} (or logic component), having a discrete behaviour and reading/writing data at each time instant, and a \emph{probabilistic evolution function}, which models the effects of the environment (physical component) on data between two time instants.
Then, $\cstep$ is defined by combining the effects on data of the program with those of the evolution function.

A typical scenario can be the following and shows that the discrete-time and the Markov assumptions are fully motivated.

\begin{notation}
In the examples throughout the paper we will slightly abuse of notation and use a variable name $x$ to denote all: the variable $x$, the (possible) function describing the evolution in time of the values assumed by $x$, and the (possible) random variable describing the distribution of the values that can be assumed by $x$ at a given time.
The role of the variable name $x$ will always be clear from the context.
\end{notation}

\begin{figure}[tbp]
\centering
\scalebox{0.7}{
\begin{tikzpicture}
\draw[-](0,4)--(0,0.5);
\draw[-](0,0.5)--(9.7,0.5);
\draw[-](10.7,0.5)--(11,0.5);
\draw[-](2.5,4)--(2.5,1);
\draw[-](2.5,1)--(3.5,1);
\draw[-](3.5,1)--(3.5,4);
\draw[-](6,4)--(6,1);
\draw[-](6,1)--(7,1);
\draw[-](7,1)--(7,4);
\draw[-](9.5,4)--(9.5,1);
\draw[-](9.5,1)--(9.7,1);
\draw[-](10.7,1)--(11,1);
\node at (1.25,0){Tank 1};
\node at (4.75,0){Tank 2};
\node at (8.25,0){Tank 3};
\draw[-latex](2.75,0.6)--(3.25,0.6);
\node at (3,0.8){$q_{12}$};
\draw[-latex](6.25,0.6)--(6.75,0.6);
\node at (6.5,0.8){$q_{23}$};
\draw[-latex](10.8,0.6)--(11.3,0.6);
\node at (11.05,0.8){$q_0$};
\draw[dashed,thick,red](-0.1,3.5)--(2.5,3.5);
\node at (-0.2,3.7){\textcolor{red}{$l_M$}};
\draw[dashed,thick,red](-0.1,0.6)--(2.5,0.6);
\node at (-0.2,0.8){\textcolor{red}{$l_m$}};
\draw[dashed,thick,red](3.5,3.5)--(6,3.5);
\draw[dashed,thick,red](3.5,0.6)--(6,0.6);
\draw[dashed,thick,red](7,3.5)--(9.5,3.5);
\draw[dashed,thick,red](7,0.6)--(9.5,0.6);
\draw[dashed,thick,ForestGreen](0,2.5)--(2.6,2.5);
\node at (2.8,2.7){\textcolor{ForestGreen}{$l_{g}$}};
\draw[dashed,thick,ForestGreen](3.5,2.5)--(6,2.5);
\draw[dashed,thick,ForestGreen](7,2.5)--(9.5,2.5);
\draw[blue](-0.1,3)--(2.5,3);
\node at(-0.2,3.2){\textcolor{blue}{$l_1$}};
\draw[blue](3.4,2)--(6,2);
\node at(3.2,2.2){\textcolor{blue}{$l_2$}};
\draw[blue](6.9,2.2)--(9.5,2.2);
\node at(6.8,2.4){\textcolor{blue}{$l_3$}};
\draw[-](0.8,6)--(0.8,5.5);
\draw[-](1.7,6)--(1.7,5.5);
\draw[-](0.2,5.5)--(2.3,5.5);
\draw[-](0.2,5.5)--(0.2,4.5);
\draw[-](2.3,5.5)--(2.3,4.5);
\draw[-](0.2,4.5)--(2.3,4.5);
\draw[-](0.8,4.5)--(0.8,4);
\draw[-](1.7,4.5)--(1.7,4);
\node at (1.25,5){pump};
\draw[-latex](1.25,4.25)--(1.25,3.75);
\node at (1.45,4.1){$q_1$};
\draw[->](3.5,5)--(2.5,5);
\node at (3.5,5.8){\scalebox{0.85}{Flow rate under}};
\node at (3.5,5.5){\scalebox{0.85}{the control of}};
\node at (3.5,5.2){\scalebox{0.85}{the software}};
\draw[-](7.8,6)--(7.8,4);
\draw[-](8.7,6)--(8.7,4);
\draw[-latex](8.25,4.25)--(8.25,3.75);
\node at (8.45,4.1){$q_2$};
\draw[->](6.5,5)--(7.5,5);
\node at (6.5,5.8){\scalebox{0.85}{Flow rate under}};
\node at (6.5,5.5){\scalebox{0.85}{the control of}};
\node at (6.5,5.2){\scalebox{0.85}{the environment}};
\draw[-](9.7,0.3)--(9.7,1.2);
\draw[-](10.7,0.3)--(10.7,1.2);
\draw[-](9.7,0.3)--(10.7,0.3);
\draw[-](9.7,1.2)--(10.7,1.2);
\node at (10.2,0.75){pump};
\draw[->](10.2,2.3)--(10.2,1.3);
\node at (10.8,3.1){\scalebox{0.85}{Flow rate under}};
\node at (10.8,2.8){\scalebox{0.85}{the control of}};
\node at (10.8,2.5){\scalebox{0.85}{the software}};
\draw[fill,cyan,opacity=0.2] (1,5.5) rectangle (1.5,6);
\draw[fill,cyan,opacity=0.2] (1.15,4.1) rectangle (1.35,4.5);
\draw[fill,cyan,opacity=0.2] (8,4.1) rectangle (8.5,6);
\draw[fill,cyan,opacity=0.2] (0,0.5) rectangle (2.5,3);
\draw[fill,cyan,opacity=0.2] (3.5,0.5) rectangle (6,2);
\draw[fill,cyan,opacity=0.2] (7,0.5) rectangle (9.5,2.2);
\draw[fill,cyan,opacity=0.2] (2.5,0.5) rectangle (3.5,1);
\draw[fill,cyan,opacity=0.2] (6,0.5) rectangle (7,1);
\draw[fill,cyan,opacity=0.2] (9.5,0.5) rectangle (9.7,1);
\draw[fill,cyan,opacity=0.2] (10.7,0.5) rectangle (11,0.8);
\end{tikzpicture}
}
\caption{Schema of the three-tanks scenario.}
\label{fig:threetanks}
\end{figure}
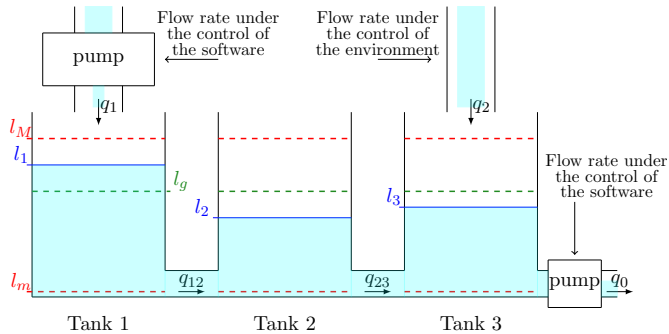

\begin{example}
\label{ex:scenario}
As outlined in the Introduction, as an example of application we consider a variant of the three-tanks laboratory experiment from \cite{RKOMI97}.
As schematised in Figure~\ref{fig:threetanks}, there are three identical tanks connected by two pipes.
Water enters in the first and in the last tank by means of 
a pump and an incoming pipe, respectively.  
The last tank is equipped with an outlet pump. 
We assume that water flows through the incoming pipe with a rate that is determined by the environment, whereas the flow rate through the two pumps is under the control of a software component. 
The task of the system consists in guaranteeing that the levels of water in the three tanks fulfil some given requirements.

The level of water in tank $i$ at time $\tau$ is denoted by $l_i(\tau)$, for $i=1,2,3$, and is always in the range $[l_{m},l_{M}]$, for suitable $l_{m}$ and $l_{M}$ giving, respectively, the minimum and maximum level of water in the tanks. 
The dynamics of $l_i(\tau)$ can be modelled via the following set of stochastic difference equations, with sampling time interval $\Delta\tau =1$:
\begin{equation}
\label{eq:tankstynamics}
\begin{array}{rcl}
l_1(\tau+1) & = & l_1(\tau) + q_1(\tau) - q_{12}(\tau)\\
l_2(\tau+1) & = & l_2(\tau) + q_{12}(\tau) - q_{23}(\tau)\\
l_3(\tau+1) & = & l_3(\tau) + q_2(\tau) + q_{23}(\tau) - q_0(\tau)
\end{array}    
\end{equation}
where $q_1$ denotes the flow rate of the pump connected to the first tank, $q_2$ denotes the flow rate of the incoming pipe, $q_{ij}$ denotes the flow rate from tank $i$ to tank $j$, and $q_0$ denotes the flow rate of the outlet pump. 
Note that $q_{12}$ and $q_{23}$ depend on $l_i$ and on the physical dimensions of the three tanks. 
We omit here all the details on the evaluation of the flow rates $q_{12}$ and $q_{23}$, that are discussed in \cite{RKOMI97} 
and reported in Appendix~\ref{app:flow_rate_details}.
We assume that the flow rate $q_2$ is under the control of the environment, so that its value is affected by the uncertainties and, thus, can only be described probabilistically.
Conversely, the two pumps are controlled by a software component that, by reading the values of $l_i(\tau)$ can select the value of $q_1(\tau+1)$ and $q_0(\tau+1)$.
The three rates assume values in the range $[0,q_{M}]$, for a given maximal flow rate $q_{M}$.
The exact value of $q_2(\tau)$ is unknown and will be evident only at execution time. 
However, we can consider different scenarios that render the assumptions we have on the environment. 
For instance, we can assume that the flow rate of the incoming pipe is normally distributed with mean $q_{av}$ and variance $\delta_q$:
\begin{equation}
\label{eq:eq_q1_option1}
\begin{array}{rcl}
q_2(\tau+1) & \sim & N(q_{av},\delta_q)
\end{array}    
\end{equation}
In a more elaborated scenario, we could assume that $q_2$ varies at each step by a value $v$ that is normally distributed with mean $0$ and variance $1$. 
In this case, we have:
\begin{equation}
\label{eq:eq_q1_option2}
\begin{array}{rcl}
v(\tau) & \sim & N(0,1) \\
q_2(\tau+1) & = & \min\left\{ \max\left\{ 0 , q_2(\tau)+v(\tau) \right\} , q_{M} \right\}
\end{array}    
\end{equation}

The equations describing the behaviour of the controller governing the behaviour of the two pumps is the following:
\begin{equation}
\label{eq:eq_q2}
q_1(\tau+1) \! = \! \left\{ \!\!
\begin{array}{ll}
\max\{ 0 , q_1(\tau) - q_s \} & \! \text{if } l_1(\tau)>l_{g}+\delta_l \\ 
\min\{ q_{M} ,  q_1(\tau) + q_s \} & \! \text{if } l_1(\tau)<l_{g}-\delta_l\\
q_1(\tau) & \! \text{otherwise.}
\end{array}
\right.
\end{equation}
\begin{equation}
\label{eq:eq_q0}
q_0(\tau+1) \! = \! \left\{ \!\!
\begin{array}{ll}
\min\{ q_{M} , q_0(\tau) + q_s \} & \! \text{if } l_3(\tau)>l_{g}+\delta_l \\ 
\max\{ 0 ,  q_0(\tau) - q_s \} & \! \text{if } l_3(\tau)<l_{g}-\delta_l\\
q_0(\tau) & \! \text{otherwise.}
\end{array}
\right.
\end{equation}
Above, $l_{g}$ is the desired level of water in the tanks, while $q_s$ is the variation of the flow rate that is controllable by the \emph{pump}. 
Moreover, $\delta_l$ is a threshold on the read level of water. 
The idea behind Equation~\eqref{eq:eq_q2} is that when $l_1$ is greater than $l_{g}+\delta_l$, the flow rate of the pump is decreased by $q_s$. 
Similarly, when $l_1$ is less than $l_{g}-\delta_l$, that rate is increased by $q_s$. 
The reasoning for Equation~\eqref{eq:eq_q0} is symmetrical.
In both cases, the use of the threshold $\delta_l$ prevents continuous contrasting updates.

The \dataspace{} for the considered system is then defined on the set of variables $\Var=\{l_1,l_2,l_3,q_1,q_2,q_0\}$ and the $\cstep$ function determining its behaviour is derived from Equations~\eqref{eq:tankstynamics},~\eqref{eq:eq_q1_option1} (or~\eqref{eq:eq_q1_option2}),~\eqref{eq:eq_q2} and~\eqref{eq:eq_q0} in the obvious way.
\end{example}


\section{The \spell metric}
\label{sec:metriche}

Our aim is now to introduce a distance measuring the differences in the behaviour of systems that will be used to define the robustness of \logicShort{} specifications.
As the behaviour of a system is totally expressed by its \traccione{}, it is natural to use the \emph{\spell metric} of \cite{CLT21}.
The definition of the hemimetric in \cite{CLT21} is based on the observation that, in most applications, the tasks of the system can be expressed in a purely data-driven fashion.
At any time step, any difference between the desired value of some parameters of interest and the data actually obtained can be interpreted as a flaw in systems behaviour.
Hence, we can introduce a \emph{penalty function} $\rho \colon \D \to [0,1]$, i.e., a continuous function that assigns to each \datastate{} $\ds$ a penalty in $[0,1]$ expressing how far the values of the parameters of interest in $\ds$ are from their desired ones (hence $\rho(\ds) = 0$ if $\ds$ respects all the parameters). 
For instance, the penalty function can be though of as a linear function assigning a value in $[0,1]$ to data states that is directly proportional to the (Euclidean) distance between the values of the parameters in the data state and their optimal value.
However, please bear in mind that this is just one possibility, as the actual definition of the penalty function depends only on the application context.

\begin{example}
\label{ex:penalty}
In the three-tanks scenario from Example~\ref{ex:scenario}, a requirement on system behaviour can be that each $l_{i}$ should be at the level $l_{g}$.
Hence, we can define penalty functions $\rho^i$, for $i=1,2,3$, as the normalised distance between the current level of water $\ds(l_i)$ and $l_{g}$, namely:
\begin{equation}
    \label{eq:single_penalty_function}
    \rho^{i}(\ds) = 
    \frac{|\ds(l_i) - l_{g}|}{\max\{l_{M}-l_{g},l_{g}-l_{m}\}}
\end{equation}
\end{example}

We can then use a penalty function $\rho$ to obtain a \emph{distance on \datastates}, namely a $1$-bounded \emph{hemimetric} $m^\D_{\rho}$: given the \datastates{} $\ds_1$ and $\ds_2$, $m^\D_{\rho}(\ds_1,\ds_2)$ expresses how much $\ds_2$ is \emph{worse} than $\ds_1$ according to parameters of interest, and thus according to $\rho$.
Since some parameters can be time-dependent, so is $\rho$: at any time step $\tau$, the $\tau$-penalty function $\rho_\tau$ compares the \datastates{} with respect to the values of the parameters expected at time $\tau$.

\begin{definition}
[Metric on \datastates{}]
\label{def:metric_DS}
For any time step $\tau$, let $\rho_{\tau} \colon \D \rightarrow [0,1]$ be the $\tau$-penalty function on $\D$.
The $\tau$-\emph{metric on \datastates{}} in $\D$, $m_{\rho,\tau}^{\D} \colon \D \times \D \to [0,1]$, is defined, for all $\ds_1,\ds_2 \in \D$, by $m_{\rho,\tau}^{\D}(\ds_1,\ds_2) = \max\{\rho_{\tau}(\ds_2)-\rho_{\tau}(\ds_1), 0 \}$.
\end{definition}

\begin{proposition}
\label{prop:m_metrica}
Function $m_{\rho,\tau}^{\D}$ is a 1-bounded hemimetric on $\D$.
\end{proposition}

Notice that $m_{\rho,\tau}^{\D}(\ds_1,\ds_2) > 0$ if and only if $\rho_\tau(\ds_2) > \rho_\tau(\ds_1)$, i.e., the penalty assigned to $\ds_2$ is higher than that assigned to $\ds_1$.
For this reason, we say that $m_{\rho,\tau}^{\D}(\ds_1,\ds_2)$ expresses \emph{how worse} $\ds_2$ is than $\ds_1$ with respect to the objectives of the system.

By means of the Wasserstein distance (cf.\ Definition~\ref{def:Wasserstein}), we can lift $m^{\D}_{\rho,\tau}$ to a hemimetric $\Wasserstein(m^{\D}_{\rho,\tau})$ over distributions in $\distrib(\D,\borel_\D)$.
The \spell hemimetric of \cite{CLT21} is then obtained as a \emph{weighted infinity norm} of the tuple of the Wasserstein distances between the distributions in the \tracciones.
As in most applications the changes on data induced by the systems can be appreciated only along wider time intervals than a computation step by the logical component (like, e.g., in the case of the evolution of the temperature in a room), a \emph{discrete, finite} set $\OT$ of time steps at which the modifications on data give us useful information on the evolution of the system is considered.

As weight we consider a non-increasing function $\lambda \colon \OT \to (0,1]$ allowing us to express how much the distance at time $\tau$ affects the overall distance between two systems.
Following the terminology used for behavioural metrics \cite{AHM03,DGJP04,CLT20}, we refer to $\lambda$ as to the \emph{discount function}, and to $\lambda(\tau)$ as to the \emph{discount factor at time} $\tau$.

\begin{definition}
[Evolution metric]
\label{def:spell_metric}
Assume a finite set $\OT$ of observation times and a discount function $\lambda$.
For each $\tau \in \OT$, let $\rho$ be a penalty function and let $m^\D_{\rho,\tau}$ be the $\tau$-metric on \datastates{} defined on it.
Then, the $\lambda$-\emph{\spell metric} over $\rho$ and $\OT$, is the mapping $\m^{\lambda}_{\rho,\OT} \colon \System \times \System \to [0,1]$ defined, for all systems $\system_1,\system_2$, by
\[
\m^{\lambda}_{\rho,\OT}(\system_1,\system_2) = \max_{\tau \in \OT}\, \lambda(\tau) \cdot \Wasserstein(m^\D_{\rho,\tau})(\ES_{\system_1,\tau}, \ES_{\system_2,\tau}).
\]
\end{definition}

\begin{proposition}
Function $\m^{\lambda}_{\rho,\OT}$ is a 1-bounded hemimetric on $\System$.
\end{proposition}

Notice that if $\lambda$ is a \emph{strictly} non-increasing function, then it specifies how much the distance of \emph{future events} is mitigated, and it guarantees that to obtain upper bounds on the \spell metric only a \emph{finite} number of observations is needed.
Hence, the choice of having $\OT$ finite is not too restrictive.


\section{The Evolution Temporal Logic}
\label{sec:requirements}

In this section we introduce the \emph{\logicName{}} (\logicShort), which allows us to specify requirements on \tracciones, and thus the properties of systems behaviour under the presence of uncertainties.

The logic bases on two atomic properties, $\sinistra{\mu}^\rho_{p}$ and $\destra{\mu}^\rho_{p}$, where $\mu$ is a distribution over \datastates{} in $(\D,\borel_\D)$, $\rho$ is a penalty function and $p$ is a real in $[0,1]$.
Informally, $\sinistra{\mu}^\rho_p$ can be used to express a desirable behaviour, whereas $\destra{\mu}^\rho_p$ can be used for unwanted, or hazardous, behaviours.
These formulae are evaluated on a distribution $\ES_{\system, \tau}$ in the \traccione{} of a system $\system$.
Let us analyse the formula $\sinistra{\mu}^{\rho}_p$ in detail.
In this case, $\mu$ is the \emph{desired} distribution over data states.
So, to establish whether the system exhibits a proper behaviour, we compare $\mu$ with the distribution $\ES_{\system, \tau}$ obtained by the system: our means of comparison is the Wasserstein lifting of the hemimetric between data states evaluated with respect to the penalty $\rho$.
(Notice that $\rho$ is a parameter of the formula $\sinistra{\mu}^{\rho}_p$.
This is due to the fact that, clearly, the penalty is not a property of the system but part of the requirements imposed on its behaviour.)
As $\mu$ is our target distribution, it is natural to check whether $\ES_{\system, \tau}$ is \emph{worse than} $\mu$, i.e., to evaluate the distance $\Wasserstein(m^\D_{\rho,\tau})(\mu,\ES_{\system,\tau})$.
Clearly, given the presence of uncertainties, it would not be feasible to say that the system satisfies the considered formula if and only if $\Wasserstein(m^\D_{\rho,\tau})(\mu,\ES_{\system,\tau}) = 0$.
Instead, we use the parameter $p$ as a tolerance on the distance: if $\ES_{\system,\tau}$ is such that $\Wasserstein(m^\D_{\rho,\tau})(\mu,\ES_{\system,\tau}) \le p$, then the behaviour of the system can be considered acceptable.
In other words, $p$ is the \emph{maximal acceptable hemi-distance} between the desired behaviour $\mu$ and the current behaviour $\ES_{\system,\tau}$.

Conversely, in the formula $\destra{\mu}^\rho_p$ the distribution $\mu$ expresses some \emph{unwanted}, \emph{hazardous}, behaviour.
Hence, the distribution $\ES_{\system,\tau}$ reached by the system must be better than $\mu$, i.e., $\Wasserstein(m^\D_{\rho,\tau})(\ES_{\system,\tau},\mu) > 0$.
Also in this case, due to the presence of uncertainties, we need to make use of a threshold parameter $p$: assuming a distribution $\ES_{\system,\tau}$ acceptable when it is only \emph{slightly} better than $\mu$ can still lead to an unwanted behaviour (because, in this case, the difference between the two distributions may only be due to some noise).
Hence, we let $p$ be the \emph{minimal required hemi-distance} between $\ES_{\system,\tau}$ and $\mu$, so that $\ES_{\system,\tau}$ is an acceptable behaviour if and only if $\Wasserstein(m^\D_{\rho,\tau})(\ES_{\system,\tau},\mu) \ge p$.

Let $\var{\mu} \subseteq \Var$ be the set of data variables over which the distribution $\mu$ is defined.
Similarly, for a penalty function $\rho$, we can consider the set $\var{\rho} \subseteq \Var$.

\begin{definition}[\logicShort]
The modal logic \logicShort{} consists in the set of formulae $\logicSymbol$ defined by the following syntax:
\[
\begin{array}{rcl}
\varphi & ::= & \top \ \;\;  |  \; \; 
\sinistra{\mu}^\rho_{p} \ \;\;  |  \; \; 
\destra{\mu}^\rho_{p} \ \;\;  |  \\
& & 
\neg\varphi \ \;\;  |  \; \; 
\varphi \vee \varphi \ \;\;  |	\;\; \funtil{\varphi_1}{[a,b]}{\varphi_2}
\end{array}
\]
with $\varphi$ ranging over $\logicSymbol$, $\mu \in \distrib{(\D,\borel_\D)}$ a distribution over data states,
$p \in [0,1]$, $\rho$ a penalty function such that $\var{\rho} \subseteq \var{\mu}$, $p \in [0,1]$ and $[a,b]$ an interval in $\OT$.
\end{definition}

Disjunction and negation are the standard Boolean connectives, and $\funtil{\varphi_1 }{[a,b]}{\varphi_2}$ is the \emph{bounded until} operator stating that $\varphi_1$ is satisfied until, at a time in $[a,b]$, $\varphi_2$ is. 

For any penalty function $\rho$, we denote by $\logicSymbol_\rho$ the sub-class of $\logicSymbol$ with atomic propositions of the form $(\cdot)^{\rho}_{(\cdot)}$.

Formulae are evaluated over systems and observable times. 
In a \emph{quantitative semantics} approach, for a formula $\varphi$, a system $\system$, and a time instant $\tau$, the value $\sat{\varphi}{\system,\tau} \in [-1,1]$ expresses the \emph{robustness} of $\system$ with respect to $\varphi$ at time $\tau$, i.e., how much the behaviour of $\system$ at time $\tau$ can be modified either while preserving the validity of property $\varphi$ (if $\varphi$ is already satisfied), or in order to obtain it.

\begin{definition}[\logicShort{}: quantitative semantics]
For any system $\system$, time step $\tau$, and \logicShort{} formula $\varphi$, the \emph{robustness} of $\system$ with respect to $\varphi$ at $\tau$, notation $\sat{\varphi}{\system,\tau} \in [-1,1]$, is defined inductively with respect to the structure of $\varphi$ as follows:
\begin{align*}
\sat{\top}{\system,\tau} ={} & 1 \\
\sat{\sinistra{\mu}^\rho_{p}}{\system,\tau} ={} & p -  \lambda(\tau)\Wasserstein(m^\D_{\rho,\tau})(\mu, \ES_{\system,\tau})\\
\sat{\destra{\mu}^\rho_{p}}{\system,\tau} ={} & \lambda(\tau)\Wasserstein(m^\D_{\rho,\tau})(\ES_{\system,\tau},\mu) - p\\
\sat{\neg\varphi}{\system,\tau} ={} & - \sat{\varphi}{\system,\tau} \\
\sat{\varphi_1 \vee \varphi_2}{\system,\tau} ={} & \max\left\{ \sat{\varphi_1}{\system,\tau}, \sat{\varphi_2}{\system,\tau} \right\} \\
\sat{\funtil{\varphi_1 }{[a,b]}{\varphi_2}}{\system,\tau} ={} & \max_{\tau' \in [\tau + a, \tau + b]}\min
\Big\{
\begin{array}{c}
\sat{\varphi_2}{\system,\tau'}, \\
\displaystyle{\min_{\tau'' \in [\tau+a, \tau')} \sat{\varphi_1}{\system,\tau''}}
\end{array}
\Big\}.
\end{align*}
\end{definition}

Intuitively, the value $\lambda(\tau) \Wasserstein(m^\D_{\rho,\tau}) (\mu,\ES_{\system,\tau})$ quantifies the (discounted) difference between the distribution $\ES_{\system,\tau}$ reached by the system $\system$ at time $\tau$ and $\mu$. 
Hence, on the one hand the robustness $\sat{\sinistra{\mu}^\rho_{p}}{\system,\tau}$ expresses whether the distribution in the \traccione{} of $\system$ is within the maximal acceptable hemi-distance $p$ from $\mu$.
On the other hand, it also expresses how much $\ES_{\system,\tau}$ can be modified while guaranteeing that the behaviour of the system remains within the specified parameters.
Clearly, the closer $\mu$ and $\ES_{\system,\rho}$, the higher the robustness.
Similarly, $\sat{\destra{\mu}^\rho_{p}}{\system,\tau}$ quantifies the robustness of $\ES_{\system,\tau}$ with respect to $\mu$ (and $\rho$) in terms of how much $\ES_{\system,\tau}$ may \emph{get close} to $\mu$ while keeping the minimal required hemi-distance $p$.
Hence, the farther $\ES_{\system,\tau}$ and $\mu$, the higher the robustness.
The semantics of boolean connectives and bounded until is standard.
Notice that due to the potential asymmetry of our distances, it is not true in general that $\destra{\mu}^\rho_p = \neg \sinistra{\mu}^\rho_{1-p}$.

As expected, other operators can be defined as macros in our logic:
\[
\begin{array}{ll}
\varphi_1 \wedge \varphi_2 \equiv \neg (\neg\varphi_1 \vee \neg \varphi_2)
\;&\;
\varphi_1 \rightarrow \varphi_2 \equiv \neg\varphi_1 \vee \varphi_2 \\[.1cm]
\fevent{[a,b]}{\varphi} \equiv \funtil{\top}{[a,b]}{\varphi_2}
\;&\;
\fglob{[a,b]}{\varphi} \equiv \neg \fevent{[a,b]}{\neg \varphi}.
\end{array} 
\]

\begin{example}
\label{ex:properties}
\logicShort{} formulae can be used to express requirements on the three-tanks scenario of Example~\ref{ex:scenario}. 
Let $\rho^i$, for $i = 1,2,3$, be the penalty functions introduced in Example~\ref{ex:penalty}.
We can express the following two requirements:
\begin{description}
    \item[Prop1:] After an initial start up period of at most $\tau_1$ steps, for the next $\tau_2$ steps the distribution of $l_3$ is at a distance of at most $p$ from a normal distribution with mean $l_{g}$ and variance $\delta_{l}$:    
    \[
        \fevent{[0,\tau_1]}{\fglob{[0,\tau_2]}{\sinistra{l_3 \sim N(l_{g},\delta_{l})}^{\rho^3}_{p}}}. 
    \]
    
    \item[Prop2:] If in the first $\tau_1$ steps the level experienced in at least one of the three tanks gets too close, i.e., at a distance less than $p_h$, to an \emph{hazardous} normal distribution with mean $l_{M}-\varepsilon$ and variance $\varepsilon$, then in at most $\tau_2$ steps the experienced values in all tanks will be close, i.e. at a distance less than $p_s$, to the \emph{safe} expected distributions:
    \[
    \begin{array}{rcl}
        \varphi &  = & \fglob{[0,\tau_1]}{(\varphi_{h}} \rightarrow \fevent{[0,\tau_2]}{\varphi_{s})}\\[.2cm]
        \varphi_{h} & = &  
        \bigvee_{i=1}^3 \neg \destra{l_i \sim N(l_{M}-\varepsilon, \varepsilon)}^{\rho^i}_{p_h}\\[.2cm]
        \varphi_{s} & = & 
        \bigwedge_{i=1}^3 \sinistra{l_i \sim N(l_{g},\delta_{l})}^{\rho^i}_{p_{s}}.
    \end{array}
    \]      
\end{description}
We remark that neither \textbf{Prop1} nor \textbf{Prop2} can be expressed with classic probabilistic temporal logics.
\end{example}

We can show that \logicShort{} characterises the distance between systems. 
More precisely, the quantitative semantics of \logicShort{} induces a distance between systems that coincides with the symmetrisation of the hemimetric $\m^{\lambda}_{\rho}$ and is therefore a pseudometric. 
Clearly, since the \spell metric is defined in terms of a given penalty function $\rho$, it will be characterised by the distance over formulae in $\logicSymbol_\rho$.

\begin{definition}[\logicShort{} distance]
Given a penalty function $\rho$, the \emph{\logicShort{} distance} over systems $\system_1$ and $\system_2$ with respect to $\rho$ and $\OT$ is defined as
\[
\logicMetric(\system_1, \system_2) = \sup_{\varphi \in \logicSymbol_\rho, \tau \in \OT}
\left| \sat{\varphi}{\system_1,\tau} - \sat{\varphi}{\system_2,\tau} \right|.
\]
\end{definition}

Firstly, we show that the symmetrisation of $\m^{\lambda}_{\rho,\OT}$ is an upper bound to $\logicMetric$.

\begin{restatable}{lemma}{lemsxdx}
\label{lem:sx_dx}
For any penalty function $\rho$, and systems $\system_1$ and $\system_2$ we have that:
\[
\logicMetric(\system_1, \system_2)
\le 
\max(\m^\lambda_{\rho,\OT}(\system_1, \system_2), \m^\lambda_{\rho,\OT}(\system_2, \system_1)).
\]
\end{restatable}

\begin{proof}
The proof can be found in Appendix~\ref{app:lem_sx_dx}.
\end{proof}

We can provide a formula witnessing that $\logicMetric$ and the symmetrisation of $\m_{\rho,\OT}^{\lambda}$ coincide.

\begin{restatable}{lemma}{lemdistinguishingtau}
\label{lem:distinguishing_tau}
For all systems $\system_1, \system_2$ and penalty functions $\rho$, there is a formula $\varphi \in \logicSymbol_\rho$ with $|\sat{\varphi}{\system_1,\tau} - \sat{\varphi}{\system_2,\tau}| = \max\left( \m^{\lambda}_{\rho,\OT}(\system_1, \system_2), \m^\lambda_{\rho,\OT}(\system_2, \system_1) \right)$,
for some $\tau \in \OT$.
\end{restatable}

\begin{proof}
The proof can be found in Appendix~\ref{app:lem_distinguishing_tau}.
\end{proof}

From Lemma~\ref{lem:sx_dx} and Lemma~\ref{lem:distinguishing_tau} we infer that the \logicShort{} distance $\logicMetric$ and the symmetrisation of $\m^\lambda_{\rho,\OT}$ coincide.

\begin{theorem}
\label{thm:sx_dx_and_thm:distinguishing_tau}
For all systems $\system_1$ and $\system_2$ we have that:
\[
\logicMetric(\system_1, \system_2)
= 
\max(\m^\lambda_{\rho,\OT}(\system_1, \system_2), \m^\lambda_{\rho,\OT}(\system_2, \system_1)).
\]
\end{theorem}

Theorem~\ref{thm:sx_dx_and_thm:distinguishing_tau} entails the \emph{soundness} (Lemma~\ref{lem:sx_dx}) and \emph{completeness} (Lemma~\ref{lem:distinguishing_tau}) of our notion of robustness.
In particular, as a direct consequence of Theorem~\ref{thm:sx_dx_and_thm:distinguishing_tau}, we can obtain the following classic result (see, e.g., \cite{DM10}): whenever the robustness 
of a system $\system$ with respect to a formula $\varphi$ is greater than the distance between $\system$ and $\system'$, then the robustness of $\system'$ with respect to $\varphi$ is positive as well.

\begin{corollary}
\label{cor:sx_dx}
Let $\varphi$ be any formula in $\logicSymbol_\rho$, $\tau \in \OT$ and let $i \in \{1,2\}$.
Whenever $\sat{\varphi}{\system_i,\tau} \ge \max( \m^\lambda_{\rho,\OT} (\system_1, \system_2), \m^\lambda_{\rho,\OT}(\system_2, \system_1))$, then $\sat{\varphi}{\system_{3-i},\tau} \ge 0$.
\end{corollary}


\section{Statistical Model Checking}
\label{sec:statisticalmc}

In this section we present an algorithm, based on \emph{statistical model-checking}, that allows us to estimate the robustness of a system $\system$ with respect to a formula $\varphi$.
This algorithm consists of three basic elements:
\begin{inparaenum}[(i)]
\item 
a randomised procedure that, based on simulation, permits the estimation of the \traccione{} of $\system$, assuming an initial \datastate{} $\ds_\system$; 
\item 
a mechanism to estimate the Wasserstein distance between two probability distributions on $(\D,\borel_\D)$; 
\item 
a procedure that by inspecting the syntax of $\varphi$ and by using the first two components computes the robustness.
\end{inparaenum}

Due to lack of space, we only present an overview of the three steps.
The algorithms we are going to discuss are reported in Appendix~\ref{sec:algorithms}.
A Python implementation of the proposed approach is available at {\small\color{blue}\url{https://github.com/gitUltron/Ultron}}.

In Section~\ref{sec:requirements} we have introduced robustness in its general form, i.e., by presenting its evaluation in a formula \emph{at any time} $\tau$.
However, the simulation of the \traccione{} of a system $\system$ can only be done starting from an initial \datastate{} $\ds_\system$ (or at least from a given finite, discrete distribution over \datastates{}), and thus, ideally, from time $0$.
Hence, it is natural, and also common practice, to evaluate the robustness of $\system$ with respect to a formula \emph{always at time} $0$.
Clearly, the temporal operators in \logicShort{} still allow us to reason on timed requirements on \tracciones.
Therefore, we will henceforth consider, for a system $\system$ and a formula $\varphi$, the robustness $\sat{\varphi}{\system} = \sat{\varphi}{\system,0}$.


\subsection{Statistical estimation of \tracciones{}}
\label{sec:stat_estimation}

\begin{figure*}[t]
\begin{subfigure}{0.5\textwidth}
\includegraphics[scale=0.3]{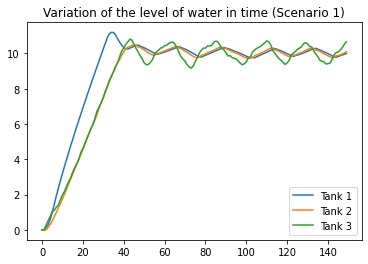}
\includegraphics[scale=0.3]{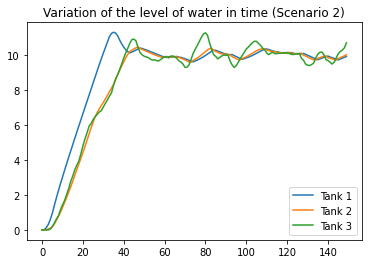}
\caption{\label{fig:simulationresult} Simulation results.}
\end{subfigure}
\hfill
\begin{subfigure}{0.55\textwidth}
\includegraphics[scale=0.3]{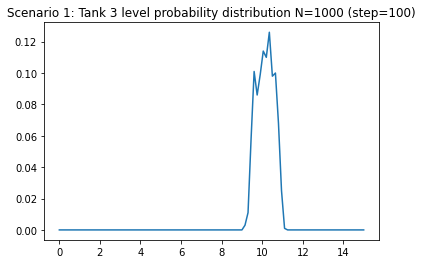}
\includegraphics[scale=0.3]{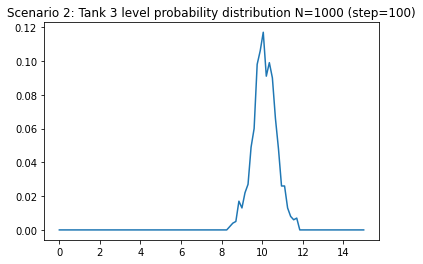}
\caption{\label{fig:probestimation} Estimated distributions of the water level in tank 3.}
\end{subfigure}
\caption{The three-tanks experiment in the two scenarios from Example~\ref{ex:scenario}.}
\end{figure*}

Given an initial \datastate{} $\ds_\system$ and an integer $k$, we let $\Call{Simul}{}$ be the function used to sample \emph{a sequence of \datastates{}} of form $\ds_0, \ds_1,\ldots, \ds_k$, modelling $k$-steps of a computation starting from $\ds_\system=\ds_0$.
$\Call{Simul}{}$ is defined assuming that for any index $0 \le i \le k-1$ and measurable set $\mesD \in \borel_{\D}$, the likelihood to sample $\ds_{i+1}\in \mesD$ after the sampling of $\ds_0, \ldots, \ds_i$ corresponds to $\cstep(\ds_i)(\mesD)$ (details in Appendix~\ref{sec:algorithms}). 

\begin{example}
\label{ex:simulation}	
A simulation of the three-tanks laboratory experiment is given in Figure~\ref{fig:simulationresult}, with the following setting:
\begin{inparaenum}[(i)]
\item $l_{m} = 0$,
\item $l_{M} = 20$,
\item $l_{g} = 10$,
\item $\delta_l = 0.5$,
\item $q_M = 6$,
\item $q_s = q_M / 5$,
\item $\delta_q = 0.5$,
\item $\Delta\tau = 0.1$.
\end{inparaenum}
On the left hand side we can see a single simulation run related to scenario 1, namely the one in which the flow rate $q_2$ of the incoming pipe is regulated by Equation~\eqref{eq:eq_q1_option1}. 
On the right hand side, we report a run of scenario 2, in which the variation of that rate is modelled as Equation~\eqref{eq:eq_q1_option2}.
In both cases, we assume an initial \datastate{} $\ds_{\system}$ with $l_i = l_m$ and $q_i = 0$, for $i = 1,2,3$.
(The parameters related to the evaluation of $q_{12}$ and $q_{23}$ in our simulations can be found in Appendix~\ref{app:flow_rate_details}.)
\end{example}

We then use a function, called $\Call{Estimate}{}$ in Appendix~\ref{sec:algorithms}, to obtain the empirical \traccione{} of system $\system$ starting from $\ds_\system$.
Intuitively, this function uses function $\Call{Simul}{}$ to obtain $N$ sampled sequences of \datastates{} $\ds_0^{j},\dots,\ds_k^{j}$, for $j=1,\dots,N$, from $\ds_\system=\ds_0^j$.
Then, a sequence of sets of samples $E_0,\ldots,E_k$ is computed, where each $E_i$ is the tuple $\ds_i^{1},\ldots,\ds_i^{N}$ of the \datastates{} observed at time $i$ in each of the $N$ sampled computations. 
Notice that, for each $i \in \{0,\dots,k\}$, the samples $\ds_i^1,\dots,\ds_i^N$ are independent and identically distributed (see Appendix~\ref{sec:algorithms}).
Each $E_i$ can be used to estimate the distribution $\ES_{\system,i}$. 
For any $i$, with $0\leq i\leq k$, we let $\hat{\ES}_{\system,i}^{N}$ be the distribution such that for any measurable set $\mesD \in \borel_{\D}$ we have
$
\displaystyle{\hat{\ES}_{\system,i}^{N}(\mesD)=\frac{|E_i \cap \mesD|}{N}}.
$
Then, by applying the weak law of large numbers to the i.i.d samples, we get that $\hat{\ES}_{\system,i}^{N}$ converges weakly to $\ES_{\system,i}$ when $N \to \infty$:
\begin{equation}
\label{eq:weak_convergence}
\lim_{N\rightarrow \infty}\hat{\ES}_{\system,i}^{N} = \ES_{\system,i}
\end{equation}

\begin{example}
\label{ex:traccione_estimation}	
In Figure~\ref{fig:probestimation} we give an estimation of the distribution of $l_3$, after $100$ time steps for $N = 1000$ samples, for the two scenarios from Example~\ref{ex:scenario}.
\end{example}

As usual in the related literature, we can apply the classic \emph{standard error approach} to analyse the approximation error of our statistical estimation of the evolution sequences.
Briefly, we let $x \in \{l_1,l_2,l_3\}$ and we focus on 
the distribution of the means of our samples (in particular we consider the cases $N = 100, 500, 1000, 5000, 10000$) for each variable $x$.
In each case, for each $i \in \{0,\dots,k\}$ we compute the mean $\overline{E_i}(x) = \frac{1}{N}\sum_{j = 1}^N \ds_i^j(x)$ of the sampled data, and we evaluate their \emph{standard deviation} $\tilde{\sigma}_{i,N}(x) = \sqrt{\frac{\sum_{j = 1}^{N} (\ds_i^j(x) - \overline{E_i}(x))^2}{N - 1}}$ (see Figure~\ref{fig:standard_deviation} for the variation in time of the standard deviation of the distribution of $x = l_3$).
From $\tilde{\sigma}_{i,N}(x)$ we obtain the \emph{standard error of the mean} $\overline{\sigma}_{i,N}(x) = \frac{\tilde{\sigma}_{i,N}(x)}{\sqrt{N}}$ (see Figure~\ref{fig:standard_error} for the variation in time of the standard error for the distribution of $x = l_3$).
Finally, we proceed to compute the $z$-\emph{score} of our sampled distribution as follows: $z_{i,N}(x) = \frac{\overline{E_i}(x) - \mathbb{E}(x)}{\overline{\sigma}_{i,N}(x)}$, where $\mathbb{E}(x)$ is the mean (or expected value) of the real distribution over $x$. 
In Figure~\ref{fig:z_score} we report the variation in time of the $z$-score of the distribution over $x = l_3$: the dashed red lines correspond to $z = \pm 1.96$, namely the value of the $z$-score corresponding to a confidence interval of the $95\%$. 
We can see that our results can already be given with a $95\%$ confidence in the case of $N = 1000$ (for readability, in Figure~\ref{fig:z_score} we have reported only the values related to $N = 100, 1000, 10000$).
Please notice that the oscillation in time of the values of the $z$-scores is due to the perturbations introduced by the environment in the simulations and by the natural oscillation in the interval $[l_{g}-\delta_l,l_{g}+\delta_l]$ of the water levels in the considered experiment (see Figure~\ref{fig:simulationresult}).
A similar analysis, with analogous results, can be carried out for the distributions of $l_1$ and $l_2$.
In Figure~\ref{fig:total_z_score} we report the variation in time of the $z$-scores of the distributions of the three variables, in the case $N = 1000$.

\begin{figure*}
\begin{subfigure}{0.48\textwidth}
\includegraphics[scale=0.35]{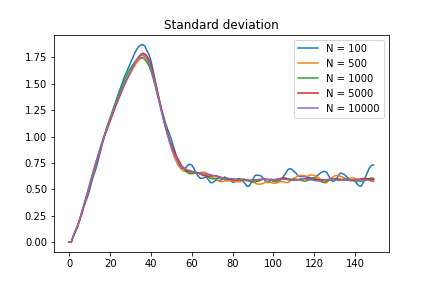}
\caption{Standard deviation for $l_3$.}
\label{fig:standard_deviation}
\end{subfigure}
\hfill
\begin{subfigure}{0.48\textwidth}
\includegraphics[scale=0.35]{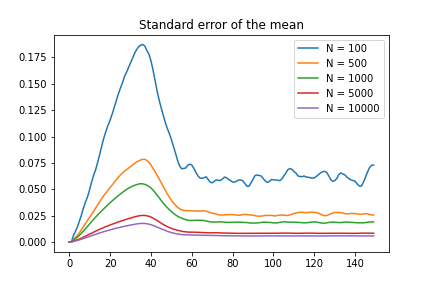}
\caption{Standard error for $l_3$.}
\label{fig:standard_error}
\end{subfigure}
\hfill
\begin{subfigure}{0.48\textwidth}
\includegraphics[scale=0.35]{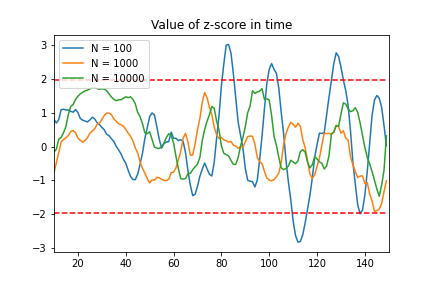}
\caption{$z$-score for $l_3$.}
\label{fig:z_score}
\end{subfigure}
\hfill
\begin{subfigure}{0.48\textwidth}
\includegraphics[scale=0.35]{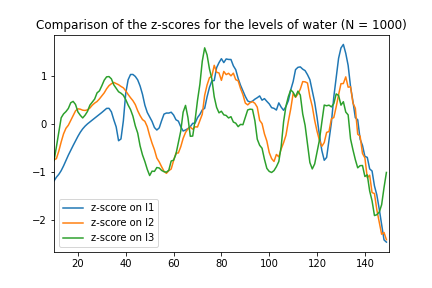}
\caption{$z$-scores for $l_1,l_2,l_3$, $N = 1000$.}
\label{fig:total_z_score}
\end{subfigure}
\caption{Analysis of the approximation error, over time, of the distributions of $l_1,l_2,l_3$, in scenario 2, for $N = 100,500,1000,5000,10000$.}
\label{fig:statistical_error}
\end{figure*}


\subsection{Statistical estimation of the Wasserstein metric}
\label{sec:computing_distance}

Let us consider two distributions $\mu$ and $\nu$ on $(\D,\borel_{\D})$.
Following an approach similar to the one presented in~\cite{TK09}, to estimate the Wasserstein distance $\Wasserstein(m^{\D}_{\rho,i})$ between (the unknown) $\mu$ and $\nu$ we can use $N$ independent samples $\{\ds^1_1,\ldots,\ds^N_1\}$ taken from $\mu$ and $\ell N$ independent samples $\{\ds^1_2,\ldots,\ds_2^{\ell N}\}$ taken from $\nu$. 
We then exploit the $i$-penalty function $\rho_i$ to map each sampled \datastate{} onto $\real$, so that it is enough to consider the sequences of values for the evaluation of the distance $\{ \omega_j=\rho_i(\ds_1^j) \}$ and $\{\nu_h=\rho_i(\ds_2^h) \}$.
We can assume, without loss of generality, that these sequences are ordered, i.e., $\omega_j\leq \omega_{j+1}$ and $\nu_{h}\leq \nu_{h+1}$. 
The value $\Wasserstein(m^{\D}_{\rho,i})(\nu,\mu)$ can be approximated as: 
\[
\frac{1}{\ell N}\sum_{h=1}^{\ell N}\max\{\nu_{h} - \omega_{\lceil \frac{h}{\ell}\rceil},0\}.
\]
The next theorem, based on results in \cite{TK09,Vil08}, ensures that the larger the number of samplings the closer the gap between the estimated value and the exact one.

\begin{restatable}{theorem}{thmestimate}
\label{thm:estimate}
Let $\mu,\nu \in \distrib(\D,\borel_\D)$ be unknown.
Let $\{\ds^1_1,\ldots,\ds^N_1\}$ be independent samples taken from $\mu$, and $\{\ds^1_2,\ldots,\ds_2^{\ell \cdot N}\}$ independent samples taken from $\nu$. 
Let $\{ \omega_j = \rho_i(\ds^j_1)\}$ and $\{ \nu_h = \rho_i(\ds^h_2)\}$ be the ordered sequences obtained by applying the $i$-penalty function to the samples.
Then, it holds, almost surely, that
\[
\Wasserstein(m^{\D}_{\rho,i})(\mu,\nu)
=
\lim_{N\rightarrow \infty} \frac{1}{\ell N}\sum_{h=1}^{\ell N}\max\left\{\nu_{h} - \omega_{\lceil \frac{h}{\ell}\rceil},0\right\}.
\]
\end{restatable}

\begin{proof}
The proof can be found in Appendix~\ref{app:thm_estimate}.
\end{proof}

In Appendix~\ref{sec:algorithms}, the function that realises the procedure outlined above is function $\Call{ComputeWass}{}$. 
Since the penalty function allows us to reduce the evaluation of the Wasserstein distance in $\real^n$ to its evaluation on $\real$, due to the sorting of $\{\nu_h \mid h \in [1,\dots,\ell N]\}$ the complexity of outlined procedure is $O(\ell N \log(\ell N))$ (cf.\ \cite{TK09}).
We refer the interested reader to \cite[Corollary 3.5, Equation (3.10)]{SFGSL12} for an estimation of the approximation error given by the evaluation of the Wasserstein distance over $N$ samples.


\subsection{Statistical estimation of robustness}

\begin{figure}
\small
\begin{algorithmic}[1]
\Function{Sat}{$\ds_\system,\varphi,\ell,N$} 
\State $k\gets \Call{Horizon}{\varphi}$
\State $E_0,\ldots, E_k\gets  \Call{Estimate}{\ds_\system,k,\ell N}$
\State \Return \Call{Eval}{$\{E_0,\ldots, E_k\},\varphi,\ell,N$}
\EndFunction
\end{algorithmic}
\caption{Function used to evaluate system robustness w.r.t.\ a formula.}
\label{alg:computesat}
\end{figure}

The computation of the robustness of a system $\system$ with respect to a formula $\varphi$, starting from the \datastate{} $\ds_\system$, is performed via the function $\Call{Sat}{}$ defined in Figure~\ref{alg:computesat}. 
Together with the \datastate{} $\ds_\system$ and the formula $\varphi$, function $\Call{Sat}{}$ takes as parameters the two integers $\ell$ and $N$ identifying the number of samplings that will be used to estimate the Wasserstein metric. 
This function consists of three steps. 
First the \emph{time horizon} $k$ of the formula $\varphi$ is computed (by induction on the structure of $\varphi$) to identify the number of steps needed to evaluate the robustness. 
In the second step, function $\Call{Estimate}{}$ is used to simulate the \traccione{} of $\system$ from $\ds_\system$ by collecting the 
sets of samplings $E_1,\ldots,E_k$ on which the robustness is computed, in the third step, by calling function $\Call{Eval}{}$ defined in Figure~\ref{alg:function_eval}.
\begin{figure}
\small
\begin{algorithmic}[1]
\Function{Eval}{$\{E_0,\ldots,E_{k}\},\varphi,\ell, N$} 
\Switch{$\varphi$}
\Case{$\top$} 
\State $\forall 0\leq i\leq k:v_i \gets 1.0$
\State \Return $v_0,\ldots,v_k$
\EndCase
\Case{$\sinistra{\mu}^\rho_{p}$}
\State $\forall 0\leq i\leq k: v_i \gets 0.0$
\State $i\gets 0$
\While{$i\leq k$}
\State $E' \gets \Call{Sample}{\mu,N}$
\State $v_i\gets p-\Call{ComputeWass}{E',E_i,\rho}$
\EndWhile
\State \Return $v_0,\ldots,v_k$
\EndCase
\Case{$\destra{\mu}^\rho_{p}$}
\State $\forall 0\leq i\leq k: v_i \gets 0.0$
\State $i\gets 0$
\While{$i\leq k$}
\State $E' \gets \Call{Sample}{\mu,\ell N}$
\State $v_i\gets \Call{ComputeWass}{E_i \downarrow N, E',\rho} - p$
\EndWhile
\State \Return $v_0,\ldots,v_k$
\EndCase
\Case{$\varphi_1\vee \varphi_2$} 
\State $v^1_0,\ldots,v^1_k=\Call{Eval}{\{E_0,\ldots,E_{k}\},\varphi_1, \ell,N}$
\State $v^2_0,\ldots,v^2_k=\Call{Eval}{\{E_0,\ldots,E_{k}\},\varphi_2, \ell,N}$
\State $\forall 0\leq i\leq k: v_i \gets \max\{ v^1_i,v^2_i \}$
\State \Return $v_0,\ldots,v_k$
\EndCase
\Case{$\neg\varphi_1$} 
\State $v^1_0,\ldots,v^1_k=\Call{Eval}{\{E_0,\ldots,E_{k}\},\varphi_1, \ell,N}$
\State $\forall 0\leq i\leq k: v_i \gets -v^1_i$
\State \Return $v_0,\ldots,v_k$
\EndCase
\Case{$\funtil{\varphi_1}{[a,b]}{\varphi_2}$}
\State $v^1_0,\ldots,v^1_k=\Call{Eval}{\{E_0,\ldots,E_{k}\},\varphi_1, \ell,N}$
\State $v^2_0,\ldots,v^2_k=\Call{Eval}{\{E_0,\ldots,E_{k}\},\varphi_2, \ell,N}$
\State \Return $\Call{Until}{v^1_0,\ldots,v^1_k~,~v^2_0,\ldots,v^2_k,a,b}$
\EndCase
\EndSwitch
\EndFunction
\end{algorithmic}	
\caption{Evaluation of robustness.}
\label{alg:function_eval}
\end{figure}
The structure of $\Call{Eval}{}$ is similar to the monitoring function for STL defined in~\cite{MN04}.
Given $\ell N$ sampled values at time $0,\ldots,k$, a formula $\varphi$ and integers $\ell$ and $N$, function $\Call{Eval}{}$ yields a tuple of the form $v_0,\ldots,v_k$, where $v_i$ is the robustness with respect to $\varphi$ at time step $i$. 
Function $\Call{Eval}{}$ is defined recursively on the syntax of $\varphi$.
If $\varphi=\top$ all the elements in the resulting tuple are equal to $1$, since $\top$ is always satisfied with 
robustness $1$. 
When $\varphi$ is $\sinistra{\mu}^\rho_{p}$ (resp. $\destra{\mu}^\rho_{p}$) the value $v_i$ is computed by first sampling $N$ (resp. $\ell N$) values of $\mu$ via the  sampling function $\Call{Sample}{}$, and then using function $\Call{ComputeWass}{}$ introduced in Section~\ref{sec:computing_distance}.
Function $\Call{Sample}{}$, given a probability distribution $\mu$ and an integer $N$, yields $N$ independent samplings of $\mu$. 
In case of $\destra{\mu}^\rho_{p}$, only $N$ elements are selected from $E_i$ (denoted by $E_i\!\downarrow\! N$).
The robustness with respect to $\varphi_1 \vee \varphi_2$ is computed as the maximum between the 
robustness with respect to $\varphi_1$ and that in $\varphi_2$. 
The robustness with respect to $\neg \varphi_1$ is computed as the additive inverse of the robustness with respect to $\varphi_1$.
Finally, when $\varphi=\funtil{\varphi_1}{[a,b]}{\varphi_2}$, the output tuple $v_0,\ldots,v_k$ is computed from the robustness with respect to $\varphi_1$ and $\varphi_2$, and the time interval $[a,b]$ via the function $\Call{Until}{}$ defined in Figure~\ref{alg:computeuntil}.

\begin{figure}
\small
\begin{algorithmic}[1]
\Function{Until}{$v^1_0,\ldots,v^1_k$~,~$v^2_0,\ldots,v^2_k$,$a$,$b$} 
\State $\forall 0\leq i\leq k: v_i \gets 0.0$
\State $i\gets 0$
\While{$i\leq k$}
\State $\forall j\in [i+a,i+b]: w^1_j=\min\{ v^1_h | h \in [i,j]\}$ 
\State $\forall j\in [i+a,i+b]: w^2_j=\min\{ w^1_j , v^2_j \}$
\State $v_i \gets \max\{ w^2_j | j\in [i+a,i+b]\}$
\EndWhile
\State \Return $v_0,\ldots,v_k$
\EndFunction
\end{algorithmic}
\caption{Evaluation of robustness with respect to the until formula.}
\label{alg:computeuntil}
\end{figure}

Given a formula $\varphi$ and an initial \datastate{} $\ds_\system$, we let $\Call{Robust}{\ds_\system,\varphi,\ell,N}=v$ if and only if $\Call{Sat}{\ds_\system,\varphi,\ell,N}=v_0,\ldots,v_k$ and $v=v_0$.
The following theorem guarantees that when $N$ goes to infinite, the robustness computed by function $\Call{Sat}{}$ converges, almost surely, to the exact value. 

\begin{restatable}{theorem}{thmrobustness}
\label{thm:robustness}
For any formula $\varphi$, system $\system$, \datastate{} $\ds_\system$, and integer $\ell>0$
\[
\lim_{N\mapsto \infty} \Call{Robust}{\ds_\system,\varphi,\ell,N} = \sat{\varphi}{\system}.
\]
\end{restatable}

\begin{proof}
The proof can be found in Appendix~\ref{app:thm_robustness}.
\end{proof}

\begin{figure*}[t]
\begin{subfigure}{.48\textwidth}
\includegraphics[scale=0.35]{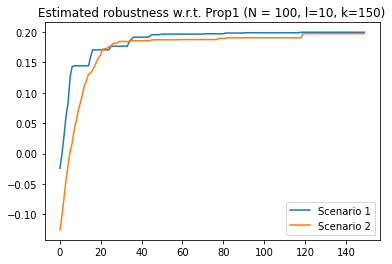}
\caption{\textbf{Prop1}, $k = 150$.}
\label{fig:mcresult_1_150}
\end{subfigure}\hfill
\begin{subfigure}{.48\textwidth}
\includegraphics[scale=0.35]{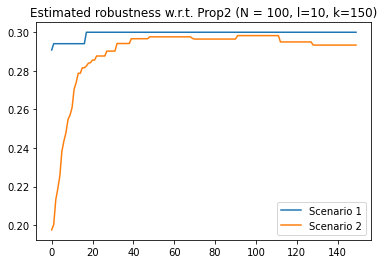}
\caption{\textbf{Prop2}, $k = 150$.}
\label{fig:mcresult_2_150}
\end{subfigure}\hfill
\begin{subfigure}{.48\textwidth}
\includegraphics[scale=0.35]{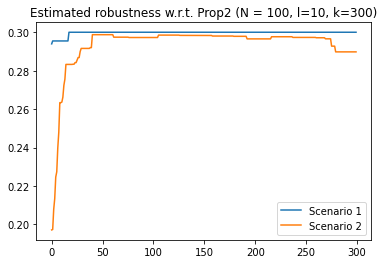}
\caption{\textbf{Prop2}, $k = 300$.}
\label{fig:mcresult_2_300}
\end{subfigure}\hfill
\begin{subfigure}{.48\textwidth}
\includegraphics[scale=0.35]{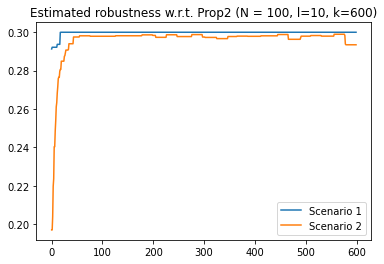}
\caption{\textbf{Prop2}, $k = 600$.}
\label{fig:mcresult_2_600}
\end{subfigure}
\caption{Estimated robustness with respect to \textbf{Prop1} and \textbf{Prop2}.}
\label{fig:mcresult}
\end{figure*}

\begin{example}
The proposed algorithm can be used to verify the requirements of Example~\ref{ex:properties}. 
In Figure~\ref{fig:mcresult} we report the results related to the following parameters instantiations: $N = 100$, $\ell=10$, and
for \textbf{Prop1} we used $\tau_1 = 20$, $\tau_2 = 30$ and $p = 0.2$;
whereas for \textbf{Prop2} we used $\tau_1 = 40, \tau_2 = 20$, $p_h = 0.2$, $p_s = 0.3$, and $\varepsilon = 0.3$.
The plots depicted in Figure~\ref{fig:mcresult} compare the variation of robustness (in time) of the system with respect to the two different scenarios from Example~\ref{ex:scenario}.
The small drop, after step $120$, in the value of the robustness with respect to \textbf{Prop2} (Figure~\ref{fig:mcresult_2_150}) is only due to a contrast between the time horizon of the formula and the time horizon $k$ of the simulation.
In fact, the same final drop can be observed when considering different time horizons for the simulation, e.g. $k = 300$ in Figure~\ref{fig:mcresult_2_300} and $k = 600$ in Figure~\ref{fig:mcresult_2_600}.
\end{example}


\section{Concluding remarks}
\label{sec:conclusion}

We have introduced the \logicName{} (\logicShort), a probabilistic variant of STL expressing requirements on the behaviour of systems under uncertainties.
Differently from the other probabilistic temporal logics usually considered in the literature, \logicShort{} can be used to express the properties of the distributions expressing the transient behaviour of the system.
Up to our knowledge, \cite{TH20} is the only other paper proposing to substitute probabilistic guarantees on the temporal properties with a richer description of the probabilistic events.
In detail, \cite{TH20} introduces \emph{ProbSTL} a stochastic variant of STL tailored to the incremental runtime verification of safe behaviour of robotic systems under uncertainties.
The objective is to develop a predictive stream reasoning tool for monitoring the runtime behaviour of robotic systems.
Hence, their stochastic signal is given by the prediction on the possible future trajectories of a system, taking into account the measurement errors by the sensors and the unpredictable behaviour of the environment.
Yet, ProbSTL specifications are tested only on the current trajectory of the system.
This is the main difference with our work, since our logic has been built to express the overall uncertain behaviour of the system.
This disparity is also a consequence of the different application context: off-line verification for us, runtime verification in \cite{TH20}.
However, as future work, we plan to develop a predictive model for the runtime monitoring of \logicShort{} specifications.
In particular, inspired by \cite{PPZGSS18,BCPSS19} where deep neural networks are used as reachability predictors for predictive monitoring, we intend to integrate our work with learning techniques, to favour the computation and evaluation of the predictions.

Another application context of probabilistic temporal logics is that of Markov processes as \emph{transformers of distributions} \cite{KA04,KVAK10}.
Roughly, one can interpret state-to-state transition probabilities as a single distribution over the state space, so that the behaviour of the system is given by the sequence of the so obtained distributions.
While this approach may resemble the \tracciones{} of \cite{CLT21}, there are some substantial differences.
Firstly, the state space in \cite{KA04,KVAK10} is finite and discrete, whereas here we are in the continuous setting.
Secondly, the transformers of distributions consider the behaviour of the system as a whole, i.e., it is not possible to separate the logical component from the environment.
Moreover, the temporal logics used to model check properties of transformers of distributions, respectively \emph{iLTL} in \cite{KA04} and the \emph{almost acyclic} B{\"u}chi automata in \cite{KVAK10}, are not comparable to our \logicShort.
In fact, those specifications have a boolean semantics, while \logicShort{} formulae are interpreted in terms of robustness.

Recently, \cite{WRWVD19} proposed a statistical model checking algorithm based on stratified sampling for the verification of PCTL specification over Markov chains.
Informally, stratified sampling allows for the generation of negatively correlated samples, i.e., samples whose covariance is negative, thus considerably reducing the number of samples needed to obtain confident results from the algorithm.
However, the proposed algorithm works under a number of assumptions restricting the form of the PCTL formulae to be checked.
While direct comparison of the two algorithms would not be feasible, nor meaningful given the disparity in the classes of formulae, it would be worth studying the use of stratified sampling in our model checking algorithm.

We also plan to investigate the application of our framework to the analysis of biological systems.
Some quantitative extensions of temporal logics have already been proposed in that setting (e.g. \cite{FR08,RBFS09,RBFS11}) to capture the notion of robustness from \cite{Ki07} or similar proposals \cite{NGM18}.
It would be interesting to see whether the use of \logicShort{} and \tracciones{} can lead to new results in this setting.

\bibliographystyle{plainurl}
\bibliography{Traccione_bib}

\newpage
\appendix

\section{Additional details on the three-tank experiment}
\label{app:flow_rate_details}

Let $A$ be the cross sectional area of each tank, and $a$ be the cross sectional area of the connecting and outlet pipes.
The volume balance difference equations of each tank, considering $\Delta\tau=1$ as sampling time interval, are the following:
\begin{equation}
\label{eq:volumes}
\begin{array}{rcl}
A(l_1(\tau+1)-l_1(\tau)) & = & q_1(\tau) - q_{12}(\tau) \\
A(l_2(\tau+1)-l_2(\tau)) & = & q_{12}(\tau) - q_{23}(\tau) \\
A(l_3(\tau+1)-l_3(\tau)) & = & q_2(\tau) + q_{23}(\tau) - q_{0}(\tau). \\
\end{array}
\end{equation}
We can then apply Torricelli's law to the equations in~\eqref{eq:volumes} to obtain the flow rates $q_{12}$ and $q_{23}$:
\begin{align*}
q_{12}(\tau) ={} & 
\begin{cases}
a_{12}(\tau) a \sqrt{2g\big(l_1(\tau) - l_2(\tau)\big)} & \text{ if } l_1(\tau)\ge l_2(\tau) \\[.1cm]
-a_{12}(\tau) a \sqrt{2g\big(l_2(\tau) - l_1(\tau)\big)} & \text{ otherwise;}
\end{cases}
\\[.2cm]
q_{23}(\tau) ={} & 
\begin{cases}
a_{23}(\tau) a \sqrt{2g\big(l_2(\tau) - l_3(\tau)\big)} & \text{ if } l_2(\tau)\ge l_3(\tau) \\[.1cm]
-a_{23}(\tau) a \sqrt{2g\big(l_3(\tau) - l_2(\tau)\big)} & \text{ otherwise;}
\end{cases}
\end{align*}
where $g$ is the gravitational constant, and $a_{12},a_{23}$ are the loss coefficients of the respective pipes.
These coefficients are represented as time dependent functions since they depend on the geometry of the pipes and on the water level in the tanks.
In our experiments (see the script at {\small\color{blue}\url{https://github.com/gitUltron/Ultron}}), we have used the approximation $g = 9.81$, and the following values for the aforementioned coefficients:
\begin{inparaenum}[(i)]
\item $a = 0.5$, 
\item $a_{12} = 0.75$,
\item $a_{23} = 0.75$.
\end{inparaenum}


\section{Proof of Lemma~\ref{lem:sx_dx}}
\label{app:lem_sx_dx}

\lemsxdx*

\begin{proof}
It is enough to show that 
\[
| \sat{\varphi}{\system_1,\tau} - \sat{\varphi}{\system_2,\tau} | 
\le 
\max(\m^\lambda_{\rho,\OT}(\system_1,\system_2), \m^\lambda_{\rho,\OT}(\system_2,\system_1))
\]
holds for all \logicShort{} formula $\varphi$ and $\tau \in \OT$.
We proceed by structural induction over $\varphi$.
\begin{itemize}
\item
Base case $\varphi \equiv \top$. Immediate.

\item
Base case $\varphi \equiv \sinistra{\mu}^\rho_{p}$.
We distinguish two cases.
The first case is
$\Wasserstein(m^\D_{\rho,\tau})(\mu, \ES_{\system_1,\tau}) \ge 
 \Wasserstein(m^\D_{\rho,\tau})(\mu, \ES_{\system_2,\tau})$. 
We have 
\begin{align*}
& 
|\sat{\varphi}{\system_1,\tau} - \sat{\varphi}{\system_2,\tau}|
\\
= \;
&
|p -  \lambda(\tau)\Wasserstein(m^\D_{\rho,\tau})(\mu, \ES_{\system_1,\tau}) -
p +  \lambda(\tau)\Wasserstein(m^\D_{\rho,\tau})(\mu, \ES_{\system_2,\tau}) |
\\
= \;
&
 \lambda(\tau)\Wasserstein(m^\D_{\rho,\tau})(\mu, \ES_{\system_1,\tau}) -
 \lambda(\tau)\Wasserstein(m^\D_{\rho,\tau})(\mu, \ES_{\system_2,\tau}) 
\\
\le \;
&
\lambda(\tau)\Wasserstein(m^\D_{\rho,\tau})(\mu, \ES_{\system_2,\tau}) 
+
\lambda(\tau)\Wasserstein(m^\D_{\rho,\tau})(\ES_{\system_2,\tau},\ES_{\system_1,\tau}) 
-
 \lambda(\tau)\Wasserstein(m^\D_{\rho,\tau})(\mu, \ES_{\system_2,\tau}) 
\\
= \;
&
\lambda(\tau)\Wasserstein(m^\D_{\rho,\tau})(\ES_{\system_2,\tau},\ES_{\system_1,\tau}) 
\\
\le \;
&
\m^\lambda_{\rho,\OT}(\system_2,\system_1)
\\
\le \;
&
\max(\m^\lambda_{\rho,\OT}(\system_1,\system_2), \m^\lambda_{\rho,\OT}(\system_2,\system_1))
\end{align*}
with the fourth step by the triangular property.
\\
The second case is $\Wasserstein(m^\D_{\rho,\tau})(\mu, \ES_{\system_1,\tau}) < 
 \Wasserstein(m^\D_{\rho,\tau})(\mu, \ES_{\system_2,\tau})$ 
and can be treated as the previous one,
by exchanging the roles of $\system_1$ and $\system_2$.

\item
Base case $\varphi \equiv \destra{\mu}^\rho_{p}$.
This case is analogous to case $\varphi \equiv \sinistra{\mu}^\rho_{p}$.

\item
Inductive step $\neg \varphi$.
We have 
\begin{align*}
& 
|\sat{\neg\varphi}{\system_1,\tau} - \sat{\neg \varphi}{\system_2,\tau}|
\\
= \;
&
|- \sat{\varphi}{\system_1,\tau} -(- \sat{\varphi}{\system_2,\tau})  |
\\
= \;
&
|\sat{\varphi}{\system_2,\tau} - \sat{\varphi}{\system_1,\tau}  |
\\
\le \;
&
\max(\m^\lambda_{\rho,\OT}(\system_1,\system_2), \m^\lambda_{\rho,\OT}(\system_2,\system_1))
\end{align*}
with the last step by the inductive hypothesis.

\item
Inductive step $\varphi_1 \vee \varphi_2$.
We distinguish two cases.\\
The first case is $\max(\sat{\varphi_1}{\system_1,\tau}, \sat{\varphi_2}{\system_1,\tau}) \ge
\max(\sat{\varphi_1}{\system_2,\tau}, \sat{\varphi_2}{\system_2,\tau})$.
\\
We have two subcases.
The first subcase is $\sat{\varphi_1}{\system_1,\tau} \ge \sat{\varphi_2}{\system_1,\tau}$.
We have:
\begin{align*}
& 
|\sat{\varphi_1 \vee \varphi_2}{\system_1,\tau} - \sat{\varphi_1 \vee \varphi_2}{\system_2,\tau}|
\\
= \;
&
|\max(\sat{\varphi_1}{\system_1,\tau}, \sat{\varphi_2}{\system_1,\tau}) - 
\max(\sat{\varphi_1}{\system_2,\tau}, \sat{\varphi_2}{\system_2,\tau})|
\\
= \;
&
\max(\sat{\varphi_1}{\system_1,\tau}, \sat{\varphi_2}{\system_1,\tau}) - 
\max(\sat{\varphi_1}{\system_2,\tau}, \sat{\varphi_2}{\system_2,\tau})
\\
= \;
&
\sat{\varphi_1}{\system_1,\tau} - \max(\sat{\varphi_1}{\system_2,\tau}, \sat{\varphi_2}{\system_2,\tau})
\\
\le \;
&
\sat{\varphi_1}{\system_1,\tau} - \sat{\varphi_1}{\system_2,\tau}
\\
\le \;
&
\max(\m^\lambda_{\rho,\OT}(\system_1,\system_2), \m^\lambda_{\rho,\OT}(\system_2,\system_1))
\end{align*}
with the last step by the inductive hypothesis.
\\
The second subcase is $\sat{\varphi_1}{\system_1,\tau} < \sat{\varphi_2}{\system_1,\tau}$ and is obtained with the same arguments, by exchanging the role of  $\varphi_1$ and  $\varphi_2$. 
\\
The second case is $\max(\sat{\varphi_1}{\system_1,\tau}, \sat{\varphi_2}{\system_1,\tau}) <
\max(\sat{\varphi_1}{\system_2,\tau}, \sat{\varphi_2}{\system_2,\tau})$ and is obtained with the same arguments of the first case, by exchanging the role of $\system_1$ and $\system_2$.

\item
Inductive case $\funtil{\varphi_1}{[a,b]}{\varphi_2}$.
We distinguish two cases.\\
The first case is $\sat{\funtil{\varphi_1}{[a,b]}{\varphi_2}}{\system_1,\tau} \ge
\sat{\funtil{\varphi_1}{[a,b]}{\varphi_2}}{\system_2,\tau}$.
\\
By definition of the semantic function $\sat{\_}{\_}$, we have that 
\[
\sat{\funtil{\varphi_1}{[a,b]}{\varphi_2}}{\system_1,\tau} =
\max_{\tau' \in [\tau + a, \tau + b]}\min\left\{ \sat{\varphi_2}{\system_1,\tau'}, \min_{\tau'' \in [\tau+a, \tau')} \sat{\varphi_1}{\system_1,\tau''} \right\},
\]
which equals the value
$\min\left\{ \sat{\varphi_2}{\system_1,\tau'''}, \min_{\tau'' \in [\tau+a, \tau''')} \sat{\varphi_1}{\system_1,\tau''} \right\}$ for a suitable $\tau'''$.
\\
We have two subcases.
The first subcase is 
$\sat{\varphi_2}{\system_2,\tau'''} \le \min_{\tau'' \in [\tau+a, \tau''')} \sat{\varphi_1}{\system_2,\tau''}$.
We have:
\begin{align*}
& 
|\sat{\funtil{\varphi_1}{[a,b]}{\varphi_2}}{\system_1,\tau} - \sat{\funtil{\varphi_1}{[a,b]}{\varphi_2}}{\system_2,\tau}|
\\
= \;
&
\sat{\funtil{\varphi_1}{[a,b]}{\varphi_2}}{\system_1,\tau} - \sat{\funtil{\varphi_1}{[a,b]}{\varphi_2}}{\system_2,\tau}
\\
= \;
&
\min\left\{ \sat{\varphi_2}{\system_1,\tau'''}, \min_{\tau'' \in [\tau+a, \tau''')} \sat{\varphi_1}{\system_1,\tau''} \right\}
 -
\max_{\tau' \in [\tau + a, \tau + b]}\min\left\{ \sat{\varphi_2}{\system_2,\tau'}, \min_{\tau'' \in [\tau+a, \tau')} \sat{\varphi_1}{\system_2,\tau''} \right\} 
\;
\\
\le  \;
&
\min\left\{ \sat{\varphi_2}{\system_1,\tau'''}, \min_{\tau'' \in [\tau+a, \tau''')} \sat{\varphi_1}{\system_1,\tau''} \right\}
 -
\min\left\{ \sat{\varphi_2}{\system_2,\tau'''}, \min_{\tau'' \in [\tau+a, \tau''')} \sat{\varphi_1}{\system_2,\tau''} \right\} 
\\
= \;
&
\min\left\{ \sat{\varphi_2}{\system_1,\tau'''}, \min_{\tau'' \in [\tau+a, \tau''')} \sat{\varphi_1}{\system_1,\tau''} \right\}
 -\sat{\varphi_2}{\system_2,\tau'''}
\\
\le \;
&
 \sat{\varphi_2}{\system_1,\tau'''}
 -\sat{\varphi_2}{\system_2,\tau'''}
\\
\le \;
&
\max(\m^\lambda_{\rho,\OT}(\system_1,\system_2), \m^\lambda_{\rho,\OT}(\system_2,\system_1))
\end{align*}
with the first step since we are in the first case, the second step from the definition of $\tau'''$, the fourth step since we are in the first subcase and the last step by the inductive hypothesis.
\\
The second subcase is $\sat{\varphi_2}{\system_2,\tau'''} > \min_{\tau'' \in [\tau+a, \tau''')} \sat{\varphi_1}{\system_2,\tau''}$. 
Let $\tau''''$ be such that 
 $\min_{\tau'' \in [\tau+a, \tau''')} \sat{\varphi_1}{\system_2,\tau''} = \sat{\varphi_1}{\system_2,\tau''''}$.
We have:
\begin{align*}
& 
|\sat{\funtil{\phi_1}{[a,b]}{\phi_2}}{\system_1,\tau} - \sat{\funtil{\phi_1}{[a,b]}{\phi_2}}{\system_2,\tau}|
\\
= \;
&
\sat{\funtil{\phi_1}{[a,b]}{\phi_2}}{\system_1,\tau} - \sat{\funtil{\phi_1}{[a,b]}{\phi_2}}{\system_2,\tau}
\\
= \;
&
\min\left\{ \sat{\varphi_2}{\system_1,\tau'''}, \min_{\tau'' \in [\tau+a, \tau''')} \sat{\varphi_1}{\system_1,\tau''} \right\}
 -
\max_{\tau' \in [\tau + a, \tau + b]}\min\left\{ \sat{\varphi_2}{\system_2,\tau'}, \min_{\tau'' \in [\tau+a, \tau')} \sat{\varphi_1}{\system_2,\tau''} \right\} 
\;
\\
\le  \;
&
\min\left\{ \sat{\varphi_2}{\system_1,\tau'''}, \min_{\tau'' \in [\tau+a, \tau''')} \sat{\varphi_1}{\system_1,\tau''} \right\}
 -
\min\left\{ \sat{\varphi_2}{\system_2,\tau'''}, \min_{\tau'' \in [\tau+a, \tau''')} \sat{\varphi_1}{\system_2,\tau''} \right\} 
\\
=  \;
&
\min\left\{ \sat{\varphi_2}{\system_1,\tau'''}, \min_{\tau'' \in [\tau+a, \tau''')} \sat{\varphi_1}{\system_1,\tau''} \right\}
 -
 \min_{\tau'' \in [\tau+a, \tau''')} \sat{\varphi_1}{\system_2,\tau''} 
\\
=  \;
&
\min\left\{ \sat{\varphi_2}{\system_1,\tau'''}, \min_{\tau'' \in [\tau+a, \tau''')} \sat{\varphi_1}{\system_1,\tau''} \right\}
 -
 \sat{\varphi_1}{\system_2,\tau''''} 
\\
\le \;
&
\sat{\varphi_1}{\system_1,\tau''''} -
\sat{\varphi_1}{\system_2,\tau''''} 
\\
\le \;
&
\max(\m^\lambda_{\rho,\OT}(\system_1,\system_2), \m^\lambda_{\rho,\OT}(\system_2,\system_1))
\end{align*}
with the first step since we are in the first case, the second step from the definition of $\tau'''$, the fourth step since we are in the second subcase, the fifth  step from the definition of $\tau''''$ and the last step by the inductive hypothesis.
\end{itemize}
\end{proof}


\section{Proof of Lemma~\ref{lem:distinguishing_tau}}
\label{app:lem_distinguishing_tau}

\lemdistinguishingtau*

\begin{proof}
Given the set of observation times $\OT$, we let
\begin{align*}
\tau_1 ={} & 
\mathrm{arg}\max_{\tau \in \OT} \lambda(\tau) \Wasserstein(m^\D_{\rho,\tau})(\ES_{\system_1,\tau},\ES_{\system_2,\tau})
\\
\tau_2 ={}& 
\mathrm{arg} \max_{\tau \in \OT} \lambda(\tau) \Wasserstein(m^\D_{\rho,\tau})(\ES_{\system_2,\tau},\ES_{\system_1,\tau})
\end{align*}
so that
\begin{align*}
& \max\left\{ \m^\lambda_{\rho,\OT}(\system_1,\system_2), \m^{\lambda}_{\rho,\OT}(\system_2,\system_1) \right\} \\
={}& 
\max\left\{ \lambda(\tau_1) \Wasserstein(m^\D_{\rho,\tau_1})(\ES_{\system_1,\tau_1},\ES_{\system_2,\tau_1}), \lambda(\tau_2) \Wasserstein(m^\D_{\rho,\tau_2})(\ES_{\system_2,\tau_2},\ES_{\system_1,\tau_2})\right\}.
\end{align*}
Assume, without loos of generality, that 
\[
\lambda(\tau_1) \Wasserstein(m^\D_{\rho,\tau_1})(\ES_{\system_1,\tau_1},\ES_{\system_2,\tau_1})
\ge
\lambda(\tau_2) \Wasserstein(m^\D_{\rho,\tau_2})(\ES_{\system_2,\tau_2},\ES_{\system_1,\tau_2}).
\]
In the symmetric case, the proof follows by switching the roles of $\system_1$ and $\system_2$.

Consider the formula
\begin{align*}
& \varphi = \sinistra{\mu}^{\rho}_p \\
& \mu = \ES_{\system_1,\tau_1} \\
& p = \lambda(\tau_2) \Wasserstein(m^\D_{\rho,\tau_2})(\ES_{\system_2,\tau_2},\ES_{\system_1,\tau_2}).
\end{align*}
Then we have
\begin{align*}
\sat{\varphi}{\system_1,\tau_1}
={}&
p - \lambda(\tau_1) \Wasserstein(m^\D_{\rho,\tau_1})(\mu,\ES_{\system_1,\tau_1})
\\
={}& 
\lambda(\tau_2) \Wasserstein(m^\D_{\rho,\tau_2})(\ES_{\system_2,\tau_2},\ES_{\system_1,\tau_2}) - 
\lambda(\tau_1) \Wasserstein(m^\D_{\rho,\tau_1})(\ES_{\system_1,\tau_1},\ES_{\system_1,\tau_1})
\\
={}&
\lambda(\tau_2) \Wasserstein(m^\D_{\rho,\tau_2})(\ES_{\system_2,\tau_2},\ES_{\system_1,\tau_2}) \\
\sat{\varphi}{\system_2,\tau_1} 
={}&
p - \lambda(\tau_1) \Wasserstein(m^\D_{\rho,\tau_1})(\mu,\ES_{\system_2,\tau_1})
\\
={}&
\lambda(\tau_2) \Wasserstein(m^\D_{\rho,\tau_2})(\ES_{\system_2,\tau_2},\ES_{\system_1,\tau_2}) - 
\lambda(\tau_1) \Wasserstein(m^\D_{\rho,\tau_1})(\ES_{\system_1,\tau_1},\ES_{\system_2,\tau_1})
\end{align*}
from which we obtain
\[
|\sat{\varphi}{\system_1,\tau_1} - \sat{\varphi}{\system_2,\tau_1}|
= 
\lambda(\tau_1) \Wasserstein(m^\D_{\rho,\tau_1})(\ES_{\system_1,\tau_1},\ES_{\system_2,\tau_1})
\]
thus concluding the proof.
\end{proof}


\section{Proof of Theorem~\ref{thm:estimate}}
\label{app:thm_estimate}

\thmestimate*

\begin{proof}
We split the proof into two parts showing respectively:
\begin{equation}
\label{proof:one}
\Wasserstein(m^\D_{\rho,i})(\mu,\nu) 
= 
\lim_{N \to \infty} \Wasserstein(m^\D_{\rho,i})(\hat{\mu}^N,\hat{\nu}^{\ell N}_{\system_2,i}) 
\enspace .
\end{equation}
\begin{equation}
\label{proof:two}
\Wasserstein(m^\D_{\rho,i})(\hat{\mu}^{N},\hat{\nu}^{\ell N})
=
\frac{1}{\ell N} \sum_{h = 1}^{\ell N} \max\left\{ \nu_h - \omega_{\lceil \frac{h}{\ell} \rceil},0 \right\}
\enspace .
\end{equation}
where $\hat{\mu}^{N}$ and $\hat{\nu}^{\ell N}$ are the estimated probability distributions obtained from $\mu$ and $\nu$ by sampling $N$ and $\ell N$ values. 
\begin{itemize}
\item {\sc Proof of Equation~\eqref{proof:one}.}

We recall that the sequence $\{\hat{\mu}^{N}\}$ (resp. $\{\hat{\nu}^{N}\}$) converges weakly to $\mu$ (resp. $\nu$) (see Equation~\eqref{eq:weak_convergence}).
Moreover, we can prove that these sequences converge weakly in $\distrib(\D,\borel_\D)$ in the sense of \cite[Definition 6.8]{Vil08}.
In fact, given the $i$-ranking function $\rho_i$, the existence of a \datastate{} $\tilde{\ds}$ such that $\rho_i(\tilde{\ds}) = 0$ is guaranteed (remember that the constraints used to define $\rho_i$ are on the possible values of state variables and a \datastate{} fulfilling all the requirements is assigned value $0$).
Thus, for any $\ds \in \D$ we have that
$
m^\D_{\rho,i}(\tilde{\ds},\ds) = \max\{\rho_i(\ds) - \rho_i(\tilde{\ds}),0\} = \rho_i(\ds).
$
Since, moreover, by definition $\rho_i$ is continuous and bounded, the weak convergence of the probability measures gives
$\int_\D \rho_i(\ds) \dd(\hat{\mu}^N(\ds)) \to \int_\D \rho_i(\ds) \dd(\mu(\ds))$ and 
$\int_\D \rho_i(\ds) \dd(\hat{\nu}^{\ell N}(\ds)) \to \int_\D \rho_i(\ds) \dd(\nu(\ds))$
and thus Definition 6.8.(i) of \cite{Vil08} is satisfied.
As $\D$ is a Polish space, by \cite[Theorem 6.9]{Vil08} we obtain that 
\[
\hat{\mu}^{N} \to \mu \text{ and } \hat{\nu}^{\ell N} \to \nu \text{ implies }
\Wasserstein(m^\D_{\rho,i})(\mu,\nu) 
= 
\lim_{N \to \infty} \Wasserstein(m^\D_{\rho,i})(\hat{\mu}^{N},\hat{\nu}^{\ell N}) 
\enspace .
\]

\item {\sc Proof of Equation~\eqref{proof:two}.}

For this part of the proof we follow \cite{TK09}.
Since the ranking function is continuous, it is in particular $\borel_\D$ measurable and therefore for any probability measure $\mu$ on $(\D,\borel_\D)$ we obtain that 
$
F_{\mu,\rho_i}(r) := \mu(\{\rho_i(\ds) < r\})
$
is a well defined cumulative distribution function.
Since, moreover, we can always assume that the values $\rho_i(\ds_1^j)$ are sorted, so that $\rho_i(\ds_i^j) \le \rho_i(\ds_1^{j+1})$ for each $j = 1,\dots,N-1$, we can express the counter image of the cumulative distribution function as
\begin{equation}
\label{eq:cdf_inverse}
F^{-1}_{\hat{\mu}^{N}, \rho_i}(r) = \rho_i(\ds_1^j)
\text{ whenever } \frac{j-1}{N} < r \le \frac{j}{N}
\enspace .
\end{equation}
A similar reasoning holds for $F^{-1}_{\hat{\nu}^{\ell N}, \rho_i}(r)$.

Then, by \cite[Proposition 6.2]{FR18}, for each $N$ we have that
\[
\Wasserstein(m^\D_{\rho,i})(\hat{\mu}^{N}, \hat{\nu}^{\ell N}) = 
\int_0^1
\max
\left\{
F^{-1}_{\hat{\nu}^{\D,\ell N}, \rho_i}(r) -
F^{-1}_{\hat{\mu}^{\D,N}, \rho_i}(r),
0
\right\}
\dd r 
\enspace .
\]
Let us now partition the interval $[0,1]$ into $\ell N$ intervals of size $\frac{1}{\ell N}$, thus obtaining
\[
\Wasserstein(m^\D_{\rho,i})(\hat{\mu}^{N}, \hat{\nu}^{\ell N}) = 
\sum_{h=1}^{\ell N}
\left(
\int_{\frac{h-1}{\ell N}}^{\frac{h}{\ell N}} 
\max\left\{
F^{-1}_{\hat{\nu}^{\ell N}, \rho_i}(r) -
F^{-1}_{\hat{\mu}^{N}, \rho_i}(r),
0
\right\}
\dd r 
\right)
\enspace .
\]
From Equation~\eqref{eq:cdf_inverse}, on each interval $(\frac{h-1}{\ell N},\frac{h}{\ell N}]$ it holds that $F^{-1}_{\hat{\nu}^{N}, \rho_i}(r) = \rho_i(\ds_1^{\lceil \frac{h}{\ell} \rceil})$ and $F^{-1}_{\hat{\nu}^{\ell N}, \rho_i}(r) = \rho_i(\ds_2^h)$.
Moreover, both functions are constant on each the interval so that the value of the integral is given by the difference multiplied by the length of the interval:
\begin{align*}
\Wasserstein(m^\D_{\rho,i})(\hat{\mu}^{N}, \hat{\nu}^{\ell N}) ={} &
\sum_{h=1}^{\ell N} \frac{1}{\ell N} 
\max\left\{\rho_i(\ds_2^h) - \rho_i(\ds_1^{\lceil \frac{h}{\ell} \rceil}), 0 \right\} \\
= & \sum_{h=1}^{\ell N} \frac{1}{\ell N} \max\left\{\nu_h - \omega_{\lceil \frac{h}{\ell} \rceil}, 0 \right\}
\enspace .
\end{align*}
By substituting the last equality into Equation~\eqref{proof:one} we obtain the thesis.
\end{itemize}
\end{proof}


\section{Proof of Theorem~\ref{thm:robustness}}
\label{app:thm_robustness}

\thmrobustness*

\begin{proof}
The proof follows by induction on the structure of the formula $\varphi$, using Theorem~\ref{thm:estimate} to deal with the base cases of $\varphi = \sinistra{\mu}^\rho_p$ and $\varphi = \destra{\mu}^\rho_p$.
\end{proof}


\section{The algorithms}
\label{sec:algorithms}

\begin{figure}[t]
\begin{minipage}[t]{0.4\textwidth}
\small
\begin{algorithmic}[1]
\Function{Simul}{$\ds_\system,k$} 
\State $i\gets 0$
\State $c \gets \ds_\system$
\State $a \gets c$ 
\While{$i\leq k$}
\State $c \gets \Call{SimStep}{c}$
\State $a \gets a,c$
\State $i \gets i+1$
\EndWhile
\State \Return $a$
\EndFunction
\end{algorithmic}
\end{minipage}
\hfill
\begin{minipage}[t]{0.55\textwidth}
\small
\begin{algorithmic}[1]
\Function{ComputeWass}{$E_1,E_2,\rho$}
\State $(\ds^1_1,\ldots,\ds^N_1)\gets E_1$
\State $(\ds^1_2,\ldots,\ds_2^{\ell N})\gets E_2$
\State $\forall j: (1\leq j\leq N): \omega_j\gets\rho(\ds^j_1)$
\State $\forall h: (1\leq h\leq \ell N): \nu_h\gets\rho(\ds^h_2)$
\State re index $\{\omega_j\}$ s.t.\ $\omega_j\leq \omega_{j+1}$	
\State re index $\{\nu_h\}$ s.t.\ $\nu_h\leq \nu_{h+1}$
\State \Return $\frac{1}{\ell N}\displaystyle{\sum_{h=1}^{\ell N}\max\{ \nu_h - \omega_{\lceil \frac{h}{\ell}\rceil}, 0\}}$
\EndFunction
\end{algorithmic}
\end{minipage}
\end{figure}
\begin{figure}[t]
\small
\begin{algorithmic}[1]
\Function{Estimate}{$\ds_\system,k,N$} 
\State $\forall i:(0\leq i\leq k): E_i \gets \emptyset$ 
\State $counter\gets 0$
\While{$counter< N$}
\State $(c_0,\ldots,c_k) \gets \Call{Simul}{\ds_\system,k}$
\State $\forall i: E_i \gets E_i,c_i$
\State $counter \gets counter + 1$
\EndWhile
\State \Return $E_0,\ldots,E_k$
\EndFunction
\end{algorithmic}	
\end{figure}

Given an initial \datastate{} $\ds_\system$ and an integer $k$, we use function $\Call{Simul}{}$ to sample \emph{a sequence of \datastates{}} $\ds_0, \ds_1,\ldots, \ds_k$, modelling $k$-steps of a computation from $\ds_\system=\ds_0$. 
Each simulation step is performed by means of function $\Call{SimStep}{}$.
The exact definition of function $\Call{SimStep}{}$ depends on the kind of \emph{model} underlying the generation of the \traccione{}.
As we are abstracting from such a model, we only assume that for any $\ds \in \D$ and measurable set $\mesD \in \borel_{\D}$:
\[
\PP(\Call{SimStep}{\ds}\in \mesD)=
\cstep(\ds)(\mesD).
\]

We use function $\Call{Estimate}{}$ to obtain the empirical \traccione{} of system $\system$ starting from $\ds_\system$.
Function $\Call{Estimate}{\ds_\system,k,N}$ invokes $N$ times function $\Call{Simul}{}$ to obtain $N$ sampled sequences of \datastates{} $\ds_0^{j},\dots,\ds_k^{j}$, for $j=1,\dots,N$, from $\ds_\system=\ds_0$.
It returns the sequence of sets of samples $E_0,\ldots,E_k$, where $E_i$ is the tuple $\ds_i^{1},\ldots,\ds_i^{N}$ of the \datastates{} observed at time $i$ in each of the $N$ sampled computations.

\end{document}